\documentclass[10pt,aps,prd,
tightenlines,twocolumn,
superscriptaddress,
nofootinbib]{revtex4}
\usepackage{hyperref}

\usepackage{amsmath,amsfonts,epsfig,graphicx,amssymb,amsthm}
\usepackage{multirow}
\usepackage{makecell}
\usepackage{float}
\usepackage{slashed} 
\usepackage{subcaption}
\usepackage{xspace}
\usepackage{xcolor}
\usepackage{enumitem}
\setlist[itemize]{align=parleft,left=5pt..18pt}

\newtheorem{proposition}{Proposition}[section]

\newcommand{\pt}{\ensuremath{p_{\text{T}}}\xspace}
\newcommand{\arctanh}{\ensuremath{\mathrm{arctanh}\;}}
\newcommand{\xt}{$x$--$t$\xspace}
\newcommand{\yt}{$y$--$t$\xspace}
\newcommand{\zt}{$z$--$t$\xspace}
\newcommand{\xy}{$x$--$y$\xspace}
\newcommand{\xz}{$x$--$z$\xspace}
\newcommand{\yz}{$y$--$z$\xspace}
\newcommand{\etaphi}{$\eta$--$\phi$\xspace}
\newcommand{\lorentznetbase}{LorentzNet\textsubscript{base}\xspace}

\begin{document}
\title{Does Lorentz-symmetric design boost network performance in jet physics?}

\author{Congqiao Li}
\email{congqiao.li@cern.ch}
\affiliation{
School of Physics and State Key Laboratory of Nuclear Physics and Technology, Peking University, 100871 Beijing, China}

\author{Huilin Qu}
\affiliation{
CERN, EP Department, CH-1211 Geneva 23, Switzerland}

\author{Sitian Qian}
\affiliation{
School of Physics and State Key Laboratory of Nuclear Physics and Technology, Peking University, 100871 Beijing, China}

\author{Qi Meng}
\affiliation{
Microsoft Research Asia, 100080 Beijing, China}

\author{Shiqi Gong}
\affiliation{
Academy of Mathematics and Systems Science, Chinese Academy of Sciences, 100190 Beijing, China}

\author{Jue Zhang}
\affiliation{
Microsoft Research Asia, 100080 Beijing, China}

\author{Tie-Yan Liu}
\affiliation{
Microsoft Research Asia, 100080 Beijing, China}

\author{Qiang Li}
\affiliation{
School of Physics and State Key Laboratory of Nuclear Physics and Technology, Peking University, 100871 Beijing, China}

\begin{abstract}
In the deep learning era, improving the neural network performance in jet physics is a rewarding task as it directly contributes to more accurate physics measurements at the LHC.
Recent research has proposed various network designs in consideration of the full Lorentz symmetry, but its benefit is still not systematically asserted, given that there remain many successful networks without taking it into account.
We conduct a detailed study on the Lorentz-symmetric design. We propose two generalized approaches for modifying a network---these methods are experimented on Particle Flow Network, ParticleNet, and LorentzNet, and exhibit a general performance gain.
We also reveal that the notable improvement attributed to the ``pairwise mass'' feature in the network is due to its introduction of a structure that fully complies with Lorentz symmetry.
We confirm that Lorentz-symmetry preservation serves as a strong inductive bias of jet physics, hence calling for attention to such general recipes in future network designs.

\end{abstract}


\maketitle


\section{Introduction}
\label{sec:intro}

Recent advancements in deep learning have had a profound impact on jet physics. Many common tasks for high-energy experimentalists have reached a new performance level with the use of deep learning techniques, which is otherwise unattainable with the classical theory-inspired approaches or shallow machine learning approaches. Jet physics tasks that already have experimental applications include jet tagging~\cite{ATLAS:2019bwq,CMS:2020poo}, jet property regression~\cite{CMS:2019uxx,CMS:2022psv}, etc. (see Ref.~\cite{Karagiorgi:2021ngt} for a recent review of deep learning applications). One major advantage deep learning approaches bring is that they directly allow proceeding with low-level data as input. Since jets are clustered from a list of initial particles, most jet datasets developed for deep learning studies are based on particle records. Regarding the representation of the particle records, the point-cloud (set) representation, which guarantees the permutational invariance of these particles, has gained increasing attention since it was proposed and developed~\cite{Henrion2017NeuralMP,Komiske:2018cqr,Qu:2019gqs}.

Improving the network performance in jet physics is a rewarding task, since advanced networks can be directly applied to real physics searches at the LHC experiments and substantially improve the sensitivity of measurements. In search of such enhancement, recent interests fall in experimenting with more advanced neural network architectures borrowing from the deep learning community, e.g., the graph neural networks (GNNs)~\cite{Henrion2017NeuralMP,Moreno:2019bmu,Qu:2019gqs,Mikuni:2020wpr,Dreyer:2020brq,Ju:2020tbo,Shlomi:2020gdn,Guo:2020vvt,Atkinson:2021nlt,Konar:2021zdg,Gong:2022lye,Ma:2022bvt} and Transformers~\cite{Mikuni:2021pou,pmlr-v162-qu22b,Qiu:2022xvr}, or injecting physics knowledge into the design of the network. For the latter, exploiting inherent symmetries in jet physics is widely studied.
The basic attempts rely primarily on data preprocessing. For instance, shifting the input jet to the center of the $\eta$--$\phi$ plain, an approach devised in early jet image representation~\cite{Cogan:2014oua,deOliveira:2015xxd}, ensures the network output is invariant under boosts on the $z$ axis (collider beam direction) or rotations on the $x$--$y$ plane.
Recently, efforts have been made to propose special network structures that respect certain symmetries. These includes networks invariant under boosts on $z$ axis or rotation on the $x$--$y$ plane~\cite{Qiu:2022xvr}, rotation on the $\eta$--$\phi$ plain~\cite{Shimmin:2021pkm} (or similarly, around the jet axis~\cite{Munoz:2022gjq}), boost along the jet axis~\cite{Munoz:2022gjq}, or even under full Lorentz transformations~\cite{bogatskiy2020lorentz,Gong:2022lye}.
Among these various symmetries, the Lorentz symmetry is considered the most fundamental, with all others being recognized as its subsymmetries.

The effort to incorporate the full Lorentz-symmetric design in the neural network first appeared in the introduction of the Lorentz layer~\cite{Butter:2017cot} and Lorentz Boost Network~\cite{Erdmann:2018shi}. The Lorentz Group Network (LGN)~\cite{bogatskiy2020lorentz} was devised not long ago to be fully equivariant under the Lorentz transformation. These network designs have attracted attention, but it is unclear to the community whether such a design could bring a real benefit since there lacks a controlled experiment studying with a similar network without the symmetric design. On the other hand, a noteworthy fact is that networks proposed in recent years that show leading performance in the context of jet tagging, including ParticleNet~\cite{Qu:2019gqs}, Attention-Based Cloud Net~\cite{Mikuni:2020wpr}, Point Cloud Transformer~\cite{Mikuni:2021pou}, and Particle Transformer (ParT)~\cite{pmlr-v162-qu22b}, still do not exploit the Lorentz-symmetric design at their cores. This poses an important question to the community: \textit{does Lorentz-symmetric design boost network performance in jet physics?}
The mentioned studies indicate that we may still lack an in-depth understanding to answer this question. 
More recently, LorentzNet~\cite{Gong:2022lye} was proposed which is fully equivariant to Lorentz transformations and surpasses ParticleNet in performance. The work includes an ablation study to demonstrate the performance gain by its symmetry-preserving design. Meanwhile, PELICAN~\cite{Bogatskiy:2022czk} also exhibits remarkable performance by solely exploiting Lorentz invariance in input features. These works utilize different approaches to preserve full Lorentz symmetries, and they all yield exceptional performances, bringing the topic to the forefront.
It therefore inspires the community to undertake a systematic study to answer the question and reveal the relations of a performant network with its Lorentz-symmetric design.

In this paper, we conduct a detailed study of the ``Lorentz-symmetric network designs.''
Our approach adheres to a general paradigm: building upon the original network, we focus on only a specific part of the network, i.e., a subnetwork, ensuring it maintains invariance under full Lorentz symmetry or its subsymmetries. This approach covers most attempts to incorporate Lorentz symmetry into networks, whether it be through the use of Lorentz-invariant inputs or by introducing dedicated network modules that keep these symmetries. The outcome is consistent: a part of the network, along with all its neurons, remains invariant under some types of Lorentz transformations. We broadly conclude this approach as Lorentz-symmetric network designs.
It is worth noting that other approaches involve embedding Lorentz equivariance within the network (e.g., Refs~\cite{bogatskiy2020lorentz,Hao:2022zns}), but they are often specialized in their designs and less easily generalizable. Thus, we reserve their exploration for future studies.
In our approach, we can make the subnetwork relatively small, hence treating it as a ``patch structure'' of the baseline network. By switching the baseline networks or changing the symmetry-related properties of the patch structure, we are able to systematically study how Lorentz-symmetric network designs influence the network performance.

Under this approach, we observe a general performance gain when incorporating Lorentz-symmetric designs in the context of jet tagging.
Our study first shows that the network performance can be improved as long as our focused patch structure keeps invariance under the Lorentz transformation, %
without the need to allow the network to respect the Lorentz symmetry fully.
The studies are based on two general proposals to integrate the Lorentz-symmetric subnetwork structures into the original network, where for the original network, we consider three baseline options for generality:
Particle Flow Network (PFN)~\cite{Komiske:2018cqr}, ParticleNet~\cite{Qu:2019gqs}, and a modified LorentzNet~\cite{Gong:2022lye}.
The experiment is complemented by a series of validations, demonstrating that the observed enhancements come from adherence to more types of Lorentz subsymmetries, progressing until the full Lorentz symmetry is attained. In addition, an important conclusion drawn from our study is the recognition of Lorentz symmetry as a valuable inductive bias in jet physics. This insight can potentially benefit a variety of jet-related tasks in future network designs.

The rest of the paper is organized as follows. In Sec.~\ref{sec:lorentz}, we review the Lorentz symmetry and discuss its specific form in jet physics. In Sec.~\ref{sec:lorentz-patch}, we devise two generalized patch structures that are invariant under Lorentz symmetry and bring a performance gain, supplemented by experiments to reveal the reason for improvements. Section~\ref{sec:conclusions} concludes our main results. Section~\ref{sec:outlook} discusses some future prospects in tagger design.

\section{Lorenz symmetry}\label{sec:lorentz}

\subsection{Lorentz transformations}

In Minkowski four-dimentional spacetime $\mathbb{R}^{1,3}$, a Lorentz vector $a^\mu$ has four components $(a^0, a^1, a^2, a^3)$, which correspond to the $t$, $x$, $y$, and $z$ components. The Minkowski metric $\eta_{\mu\nu} = \mathrm{diag}(+1, -1, -1, -1)$ defines the inner product of two Lorentz vectors $a^\mu b^\nu \eta_{\mu\nu} = a^0 b^0 - a^1 b^1 - a^2 b^2 - a^3 b^3$.
Lorentz transformations are linear transformations $\Lambda^\mu_{\;\;\nu}$ that preserve the Minkowski metric: $g_{\mu\nu}\Lambda^{\mu}_{\;\;\alpha}\Lambda^{\nu}_{\;\;\beta} = g_{\alpha\beta}$. 
Hence, the inner product of two Lorentz vectors remains unchanged. The Lorentz transformations that preserve the direction of time form the orthochronous Lorentz group, $\mathrm{SO}^{+}(1, 3)$.

The infinitesimal transformations in $\mathrm{SO}^{+}(1, 3)$ include 6 degrees of freedom. From the physics interpretation, these include three types of rotation in the space dimensions (we denote them as \xy, \xz, and \yz rotation in what follows), and three types of Lorentz boosts involving the time dimension (denoted as \xt, \yt, and \zt boost). Here we consider the finite-size transformations in the mathematical form. Taking \xy rotation and \zt boost as an example, the \xy rotation is presented as
\begin{equation}
\Lambda^{\mu}_{\;\;\nu} = 
    \begin{pmatrix}
    1\quad\quad &  &  & \\
      & \cos\alpha\; & -\sin\alpha & \\
      & \sin\alpha\; & \cos\alpha & \\
      &  &  & \quad1
    \end{pmatrix},
\end{equation}
and the \zt boost has the form
\begin{equation}
\Lambda^{\mu}_{\;\;\nu} = 
    \begin{pmatrix}
    \;\cosh w & & & \sinh w\;  \\
     & \;1\;\; & & \\
     &  & \;1\;\; & \\
    \;\sinh w &  &  & \cosh w\; 
    \end{pmatrix},
\end{equation}
where $\alpha$ stands for the rotation angle and $w$ for the boost rapidity.

\subsection{Lorentz symmetry for jet physics}\label{sec:lorentz-symm-jet}

In the context of jet physics, a jet is a collinear spray of particles produced in high-energy collisions. When presenting it to the jet network, a jet is composed of a list of particles, where each particle carries the Lorentz vectors $p^\mu$---its energy-momentum vector, and some Lorentz scalars, e.g., the particle ID\footnote{For some jet datasets, each particle may also include information from trajectory displacement that, by geometry, cannot be presented in the forms of Lorentz scalars or vectors. These are not included in our study.}.
For jets appearing in the ATLAS or CMS detector at the LHC, it is conventional to define the $z$ axis pointing to the beamline direction, and the \xy plane as the transverse plane. It is an inherent aspect of hadron colliders that the physics properties of an event and the jets it produces remain unchanged when all postcollision particles undergo \zt boost and \xy rotation. Therefore, it is conventional for the output of the jet network to be invariant under these two transformations.

Additionally, for ATLAS or CMS experiments, the particle is generally considered in the relativistic limit, as the mass of the particle is on the level of $o(0.1)$ GeV, which is smaller by 1--4 orders of magnitude than its momentum or energy. For applications to feed the jet kinematics features into the deep neural network, the requirements for float number precision are not very demanding. Therefore, it is safe to make the following assumption:
\begin{equation}
    p^\mu p_\mu = 0.
\end{equation}

Important features for jet physics include pseudorapidity $\eta$ and azimuthal angle $\phi$. In the relativistic limit, we have
\begin{align}\label{eq:etaphi}
\begin{split}
    \eta &= \arctanh \frac{p_z}{E}, \\
    \phi &= \arctan \frac{p_y}{p_x}.
\end{split}
\end{align}
Note that a \zt boost by rapidity $y_z$ and an \xy rotation by angle $\alpha_z$ to a particle with $(\eta,\,\phi)$ should directly result in $(\eta',\,\phi') = (\eta + y_z,\,\phi + \alpha_z)$.

A neural network applied to the jet physics tasks is considered to preserve the Lorentz symmetry if its output score is invariant when the input jet undergoes any Lorentz transformation. In this case, the nodes of the neural network can either be invariant, which means the nodes are Lorentz scalars, or be equivariant to the transformations, meaning that they are part of the vector or high-order tensor in the Lorentz group representation. As an application of this scenario, the LGN includes nodes that are Lorentz scalars, vectors, and high-order tensors~\cite{bogatskiy2020lorentz}; meanwhile, LorentzNet is constructed by nodes only from Lorentz scalars and vectors~\cite{Gong:2022lye}.

It is also possible that the network is only invariant or equivariant to certain kinds of transformations. As Sec.~\ref{sec:intro} mentions, there are generally two means to respect certain symmetries when designing networks.
One simple approach is to use input data that are invariant to a kind of symmetry. This typically involves a data preprocessing stage before inputting the data into the network. The following discussion refers to this as the ``data engineering'' approach.
For example, particle-level features $\pt$, $\Delta \eta$, $\Delta \phi$, or $\Delta R$ are invariant under the \zt boost and \xy rotation\footnote{The defination of these variables are $\pt = (p_x^2 + p_y^2)^{\frac{1}{2}}$, $\Delta \eta = \eta - \eta_{\text{jet}}$, $\Delta \phi = \phi - \phi_{\text{jet}}$, and $\Delta R = (\Delta\eta^2 + \Delta\phi^2)^{\frac{1}{2}}$.}.
Therefore, designing any form of the neural network will maintain the invariance property of the output score to any \zt boost and \xy rotation. This implies that one can use $\Delta \eta$, $\Delta \phi$ instead of $\eta$, $\phi$ of the particle to preserve this symmetry.
We note this approach is generally adopted by most network implementations that utilize the particle-level features as input. Its origin dates back to the early convolutional neural network (CNN) approaches~\cite{Cogan:2014oua,deOliveira:2015xxd}, where a standard preprocessing is always applied to reposition the jet image on the \etaphi plane to be centered at $(0,\,0)$. In addition, it is common for these CNN methodologies to apply additional preprocessing to rotate the jet image on the \etaphi plane into a standardized orientation. Thus, the rotational symmetry on the \etaphi plane is further maintained.

In addition to the data engineering approach mentioned above, another solution is to specially design the network so that its output is invariant under a certain group of transformations. For instance, the Particle Convolution Network~\cite{Shimmin:2021pkm} introduces dedicated convolution on \etaphi space so that the symmetry under ``rotation'' on the \etaphi plane is maintained. The Covariant Particle Transformer~\cite{Qiu:2022xvr} has its Transformer block designed to be equivariant under the Lorentz \zt boost and \xy rotation.

The above facts show that many networks have considered incorporating symmetry in their design, whether implicitly or explicitly, but the question is this: do we have a systematic way to understand and categorize these symmetries? How are these symmetries related to the largest symmetry group---the orthochronous Lorentz group? We interpret it through the following theoretical analysis.

For each jet, we first deliver a \zt boost and \xy rotation to the jet to have $(\eta,\,\phi) = (0,\,0)$.
An equivalent way of understanding this operation is to perform a translation on the jet's \etaphi plane representation such that the axis of the jet points at the $(\eta,\,\phi)$ origin. Note that the jet axis is now fixed at the $x$ axis in the three-dimensional view.
Given that we have fixed 2 degrees of freedom out of 6, there are 4 additional degrees of freedom to Lorentz transform the jet. As illustrated in Fig.~\ref{fig:lorentz-vis}, the four transformations are \yz rotation, \xt boost, $z$-tilt, and $y$-tilt. The latter two are a mixture of \zt boost with \xz rotation and a mixture of \yt boost with \yz rotation. Note that the \zt boost and \xz rotation are not commutable, similar to \yt boost with \yz rotation---we adopt the convention to first deliver the boost, followed by the rotation in the following context.

\begin{figure*}[tb]
\begin{center}
\centerline{\includegraphics[width=\textwidth]{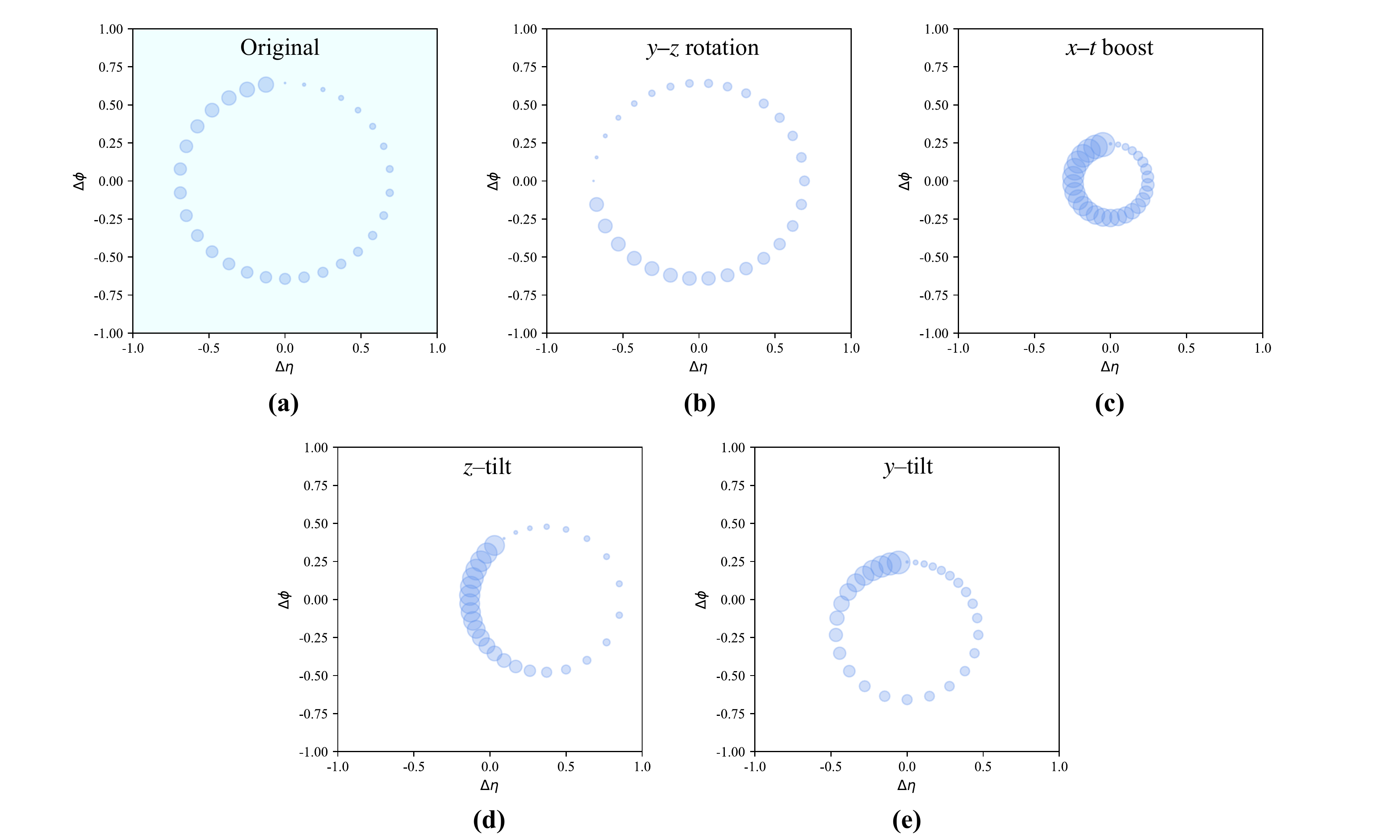}}
\caption{Illustration of a toy jet on the (a) \etaphi plane and its behavior when it undergoes the four types of Lorentz transformation that maintain the jet axis directing to the $x$ axis or, equivalently, the origin of the \etaphi plane. The four types of Lorentz transformation are (b) \yz rotation, (c) \xt boost, (d) $z$ tilt, and (e) $y$ tilt. Markers in the plot represent the constituent particles of the jet, where the size of the marker represents the \pt of the particle.}
\label{fig:lorentz-vis}
\end{center}
\end{figure*}

From Figs.~\ref{fig:lorentz-vis} (a) and \ref{fig:lorentz-vis} (b), it is then clear that the previously discussed \etaphi rotation is an approximate \yz rotation, when the jet is fixed at the origin of the \etaphi plane. The approximation holds in the limit $p_y\sim o(E)$ and $p_z\sim o(E)$ (i.e., $p_{y,z} \ll E$). Define rapidities
\begin{align}
\begin{split}
    y_y &= \arctanh \frac{p_y}{E}, \\
    y_z &= \arctanh \frac{p_z}{E},
\end{split}
\end{align}
we have $y_{y,z}\sim o(1)$.
According to Eq.~(\ref{eq:etaphi}), we have 
\begin{align}\label{eq:etaphiapp}
\begin{split}
    \eta &= \frac{p_z}{E} + o(y_z), \\
    \phi &= \frac{p_y}{E} + o(y_y, y_z).
\end{split}
\end{align}
This proves that \etaphi rotation is essentially an approximate \yz rotation. Reference~\cite{Shimmin:2021pkm} also considers the possibility of adding an invariance property on \xt boost to the network design. We reveal that they can all be grouped into our four transformation prototypes.

\section{Lorentz-symmetric patches}\label{sec:lorentz-patch}

Given the above background, we perform the studies following the aforementioned Lorentz-symmetric network design paradigm.
We isolate a specific part of the baseline network, known as a patch network structure, and make it invariant under one or some of the four transformations. The invariance is achieved by inputting invariant features under certain transformations to the patch structure. In order to isolate such a patch structure, we either put a ``patch'' to the well-established baseline network or choose a certain part of the baseline network and isolate it. The specific means are elaborated in the subsections as follows. The experiment is delivered to compare different symmetric design scenarios, evaluate if there is performance gain, and study its relation to the additional symmetries brought to the system. 

\subsection{Incorporating pairwise features}

As a starting point, we consider the scheme that the subnetwork is fully invariant under Lorentz transformations. Therefore, all inputs to the subnetwork are required to be Lorentz scalars. Given $N$ particles with Lorentz vectors $p_i^\mu$, the only possible Lorentz scalars constructed are $m_{ij}^2 = (p_i)^\mu (p_j)_\mu$. They are denoted as pairwise mass features in the following context. This approach will bring $N(N-1)/2$ pairwise features into the network. To incorporate these pairwise features, we choose GNNs as our baselines because they exploit the edge features and deliver message-passing between nodes by design. Hence, we use ParticleNet~\cite{Qu:2019gqs} and LorentzNet~\cite{Gong:2022lye} as our baseline models for the study.

In previous works (e.g., ParT~\cite{pmlr-v162-qu22b}), the pairwise mass has been studied and found to be helpful in improving network performance. In this work, we hope to go one step forward to understand the logic for such improvement. We discover that the underlying reason lies in symmetry preservation---we reveal this fact by studying many possible options to construct a pairwise variable that respects different levels of Lorentz symmetry, as collated in Sec.~\ref{sec:lorentz-symm-jet}.

\subsubsection{Variables}
The following pairwise variables are chosen in our study.

\paragraph*{Pairwise mass.}
$m_{ij} = (p_i^\mu p_{j,\mu})^{\frac{1}{2}}$ is known to be invariant under all types of Lorentz transformations.

\paragraph*{Pairwise $\Delta R$.}
$\Delta R_{ij} = (\Delta\eta_{ij}^2 + \Delta\phi_{ij}^2)^{\frac{1}{2}}$ is another physics-motivated variable that measures the angular separation of two particles. It is interesting to figure out that the variable is not only invariant under \zt boost and \xy rotation, but also invariant under rotation in the \etaphi plane, hence an approximate invariance under \yz rotation when the jet direction points to the $x$ axis.

\paragraph*{Pairwise \pt-weighted $\Delta R$.}
Inspired by the symmetry perspective, we consider a new type of angular separation angle that is further approximately invariant under the \xt boost. From Figs.~\ref{fig:lorentz-vis} (a) and \ref{fig:lorentz-vis} (c), we see that an \xt boost in the positive direction decreases the angles between particles while raising the transverse momentum \pt. It can be proved that they are in inverse proportion in the limit of $y_{y,z}\sim o(1)$ for all particles (see proof in Appendix~\ref{sec:proof}). Hence, we construct from pure mathematics the new pairwise variable, $\Delta R_{ij}(p_{\text{T},i} + p_{\text{T},j})$.

\paragraph*{Pairwise energy.}
To design the ablation experiment, the pairwise energy variables $E_{ij} = E_i + E_j$ are chosen that, in general, violate the Lorentz symmetry. They do not obey the two basic symmetries, i.e., the \zt boost and the \xy rotation, but they are actually invariant under any rotation in 3D space.

Table~\ref{tab:pairwise-vars} summarizes the four constructed variables and their invariance property under the specific types of Lorentz transformation.

\begin{table*}[tb]
\setlength{\tabcolsep}{8pt}
\caption{Invariance property of the pairwise variables between particles when the jet undergoes a certain type of Lorentz transformation.}
\label{tab:pairwise-vars}
\begin{center}
\resizebox*{1\textwidth}{!}{
\begin{tabular}{l|cccccc}
\hline\hline
Pairwise variable & \zt boost & \xy rotation & \makecell[c]{\yz rotation\\ ($y_{y,z}\sim o(1)$)} & \makecell[c]{\xt boost\\ ($y_{y,z}\sim o(1)$)} & \makecell[c]{$z$ tilt\\ ($y_{y,z}\sim o(1)$)} & \makecell[c]{$y$ tilt\\ ($y_{y,z}\sim o(1)$)} \\
\hline
$m_{ij}^2$ & \checkmark & \checkmark & \checkmark & \checkmark & \checkmark & \checkmark \\ 
$\Delta R_{ij}$ & \checkmark & \checkmark & \checkmark  & & &  \\
$\Delta R_{ij}(p_{\text{T},i} + p_{\text{T},j})$ & \checkmark & \checkmark & \checkmark & \checkmark & &  \\
\hline
$E_{ij}$ (ablation study) &  & \checkmark & \checkmark & & & \\
\hline\hline
\end{tabular}
}
\end{center}
\end{table*}

\subsubsection{Baseline networks}

Our baseline neural network should satisfy two requirements. First, the network has a GNN backbone so as to have an intrinsic mechanism to incorporate edge features. Second, the network has not, by design, included the above variable. We use ParticleNet and the ``weakened'' LorentzNet, named \lorentznetbase, as our baseline networks. A detailed description is as follows.

ParticleNet has satisfied the two requirements by default~\cite{Qu:2019gqs}. LorentzNet by design uses the pairwise masses to build edge features in each unit block. Specifically, the edge feature is constructed by concatenating Lorentz scalar node features for two connecting nodes $h_i$ and $h_j$, with the mass variables, namely, $\Vert p_i + p_j \Vert^2$ and $p_i^\mu p_{j,\mu}$~\cite{Gong:2022lye}. We simply remove the two mass variables in the construction of the edge feature. Furthermore, we find that by completing all node input variables in LorentzNet such that they are the same as the ParticleNet variables, there is no performance loss, although some variables are not Lorentz scalars, which violates the spirit of the original LorentzNet design. Adding these additional node features can, however, improve the network performance in the case where we remove the mass in edge features, which is as expected. In this way, we created a specific version of LorentzNet that is more similar to ParticleNet. The modified LorentzNet model is denoted by \lorentznetbase. Like ParticleNet, it does not hold the Lorentz-invariant or equivariant properties. 

\subsubsection{Patch structure}\label{sec:pair-patch-design}

We then introduce how to incorporate pairwise features into the baseline network.

\paragraph*{ParticleNet.} ParticleNet is built from a stack of EdgeConv~\cite{wang2019dynamic} layers that perform a ``graph convolution'' on a point cloud. It includes an intrinsic message-passing mechanism for each node with their $k$ nearest neighbors. Specifically, for each node $i$ carrying features $\mathbf{x}_i$, consider its neighboring nodes $i_j$ $(j=1,...,k)$, the message $(\mathbf{x}_{i_1},\,...,\,\mathbf{x}_{i_k})$ is passed to target node $i$. Hence, there leaves space to include manually designed pairwise features between the node $i$ and $i_j$. Figure~\ref{fig:patch} (a) illustrates the patch structure we introduce to the original ParticleNet model. To begin with, the $N(N-1)/2 $ pairwise features are calculated. Starting from initial feature dimension 1, they are embedded in latent space with dimension 64 by an elementwise multilayer perceptron (MLP), via two hidden layers both with feature dimension 64. The embedded feature is denoted by $\mathbf{U}_{ij}$ for nodes $i$ and $j$, and then it proceeds to all the EdgeConv blocks. For each EdgeConv block, the new message conveyed from neighboring nodes can be constructed by $(\mathbf{U}_{i,i_1},\,...,\,\mathbf{U}_{i,i_k})$. The feature vectors of $\mathbf{U}_{ij}$ are directly added to the original message, after passing an individual linear layer to match their dimensions. This is represented by the following equation and depicted in Fig.~\ref{fig:patch} (a):
\begin{align}\label{eq:pair-conv2d}
\begin{split}
    \mathbf{x}'_{i_j} = \mathbf{x}_{i_j} + \text{\texttt{Linear}}(\mathbf{U}_{i,i_j})\quad (j = 1,\cdots,k).
\end{split}
\end{align}

\paragraph*{\lorentznetbase.} The implementation of the pairwise feature to \lorentznetbase is much easier, as we directly adopt the intrinsic mechanism of the original LorentzNet to incorporate pairwise features. We note that there are two main differences with ParticleNet implementation. First, the pairwise features are repeated for calculation in each unit layer, as in LorentzNet, and the vectors are updated dynamically layer by layer, hence the pairwise features also change. Second, due to the fact that LorentzNet is essentially a fully connected GNN, pairwise features for all $N(N-1)/2$ pairs participate in the network.

\begin{figure*}[!tb]
\begin{center}
\begin{subfigure}[b]{\textwidth}
\centering
\includegraphics[width=\textwidth]{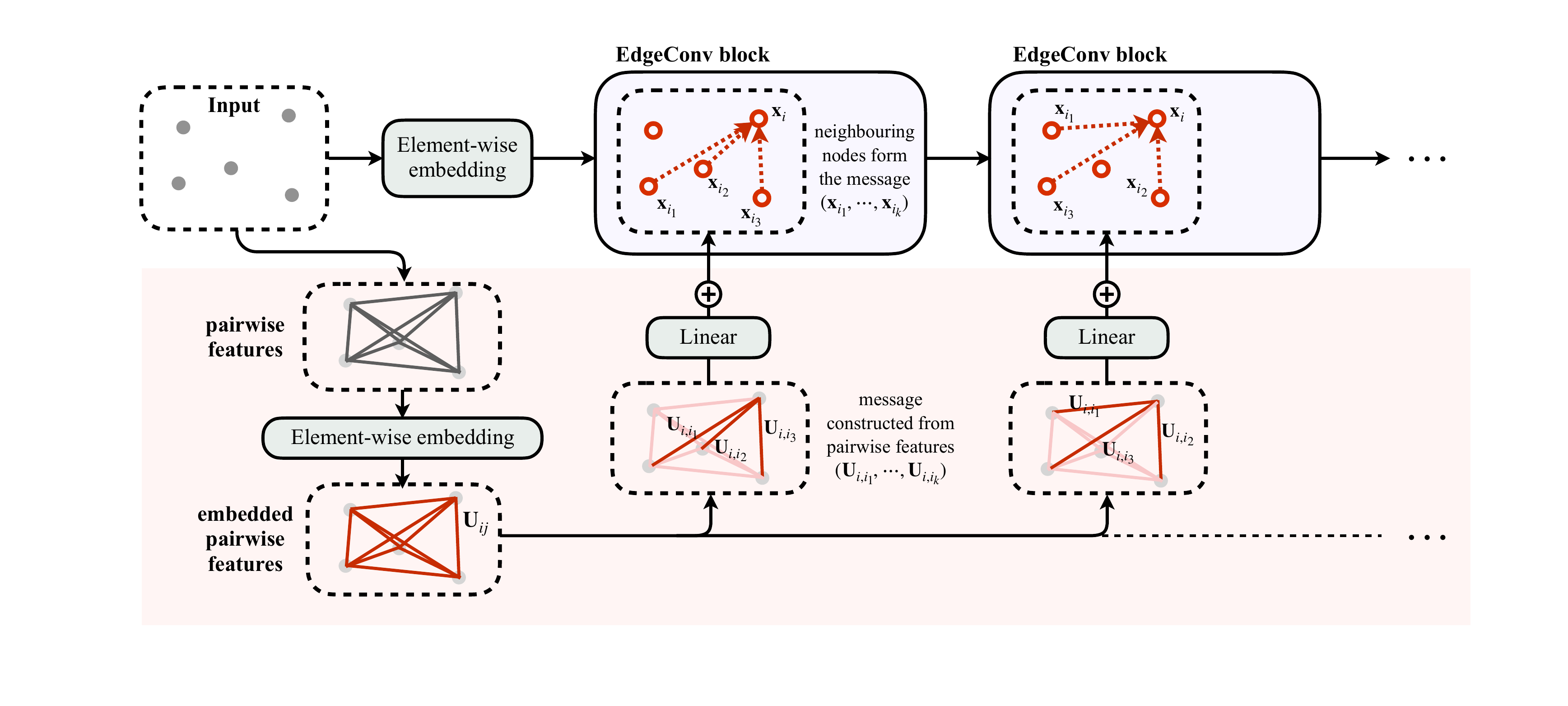}
\caption{The patch structure introduced for ParticleNet to incorporate the pairwise features of the input particles. The patch structure is drawn in a red background to be distinguished from the original ParticleNet structure. The pairwise features, after embedded, are integrated into each of the EdgeConv blocks according to the intrinsic $k$-nearest neighbours mechanism to define pairs of particles.}
\end{subfigure}
\begin{subfigure}[b]{\textwidth}
\centering
\includegraphics[width=\textwidth]{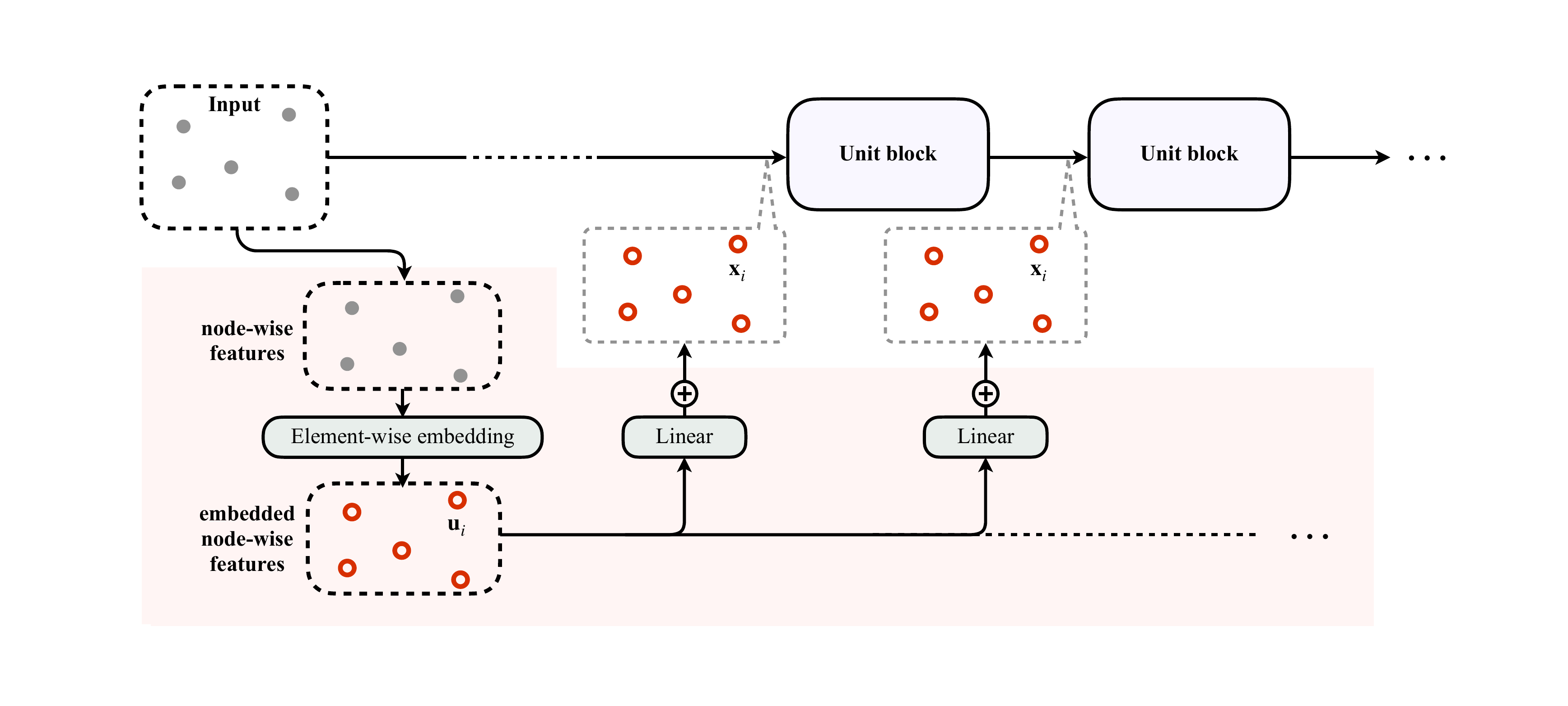}
\caption{The generalized patch structure introduced for all baseline networks to incorporate additional nodewise features. The patch structure is drawn in a red background to be distinguished from the original network structure. The nodewise features, after embedded, are integrated by ``summation'' to the latent space features fed into every unit layer of the baseline network.}
\end{subfigure}
\caption{Illustration of the patch structure introduced to the baseline networks that incorporates (a) the pairwise and (b) nodewise features.}
\label{fig:patch}
\end{center}
\end{figure*}

\subsubsection{Experiments}\label{sec:pair-exp}

The network performance is assessed in the task of jet tagging, which is a classification task to identify the origin of a jet.
Our experiments are performed on two datasets: the top tagging dataset~\cite{kasieczka2019top} and the JetClass dataset~\cite{pmlr-v162-qu22b}.
The top tagging dataset is a jet dataset comprising $1.2\times 10^{6}$ jets for training. It includes two sets of large-radius jets, one initiated from the top quark and the other from the QCD events initiated by quarks and gluons. Events in this dataset are generated by \textsc{pythia8}~\cite{Sjostrand:2014zea} and passed to \textsc{delphes}~\cite{deFavereau:2013fsa} for fast simulation of the detector effect. Jets are clustered from the E-flow objects with the anti-$k_{\rm T}$ algorithm~\cite{Cacciari:2008gp}. The kinematics information of the jet constituents is used as the only input source to train the network. The JetClass dataset is a larger jet dataset with $100\times 10^{6}$ jets for training. It is composed of ten classes of large-radius jets, including five decay modes of the Higgs boson, two decay modes of the top quark, two initiated by $Z$ and $W$ bosons, and one from the QCD event initiated by quarks and gluons. Events in this dataset are first generated by \textsc{MadGraph5\_mc@nlo}\cite{Alwall:2014hca} for resonance production and their decay, then they proceed to \textsc{pythia8}~\cite{Sjostrand:2014zea} for parton showering and \textsc{delphes}~\cite{deFavereau:2013fsa} for detector effect simulation. Jets are reconstructed similarly from the E-flow objects. The constituent-level features are used as the input to the jet network, including the kinematics information, particle identification flags, and trajectory displacement features.

In the following experiments, we use the top tagging dataset as the main benchmark dataset to study how Lorentz-symmetric network designs influence the network performance in various aspects. Particularly, we limit our training to $60\,000$ jets, as our findings indicate that for the top tagging benchmark, utilizing the entire $1.2\times 10^{6}$ jets for training leads to a saturation in network performance. This obscures the distinctions when we test with various top-performing networks and switch their subcomponents in our studies.
To validate that the conclusion is applicable to a wide range of data sizes, we perform an experiment on different sizes of the top tagging task and JetClass's ten-class classification task, covering the data size from 6000 to $100\times 10^{6}$.


The training setup is the same with Ref.~\cite{pmlr-v162-qu22b}, only with a proper resetting of the batch size to cooperate with the more complex computation when pairwise features are involved. The LorentzNet model is trained in the same optimizer and scheduler as in the ParticleNet case. A detailed description of the training setup is presented in Appendix~\ref{sec:config}.

Table~\ref{tab:pair-perf} shows the evaluation results for different network designs. A number of metrics are used for evaluation, including the accuracy, area under the receiver operating characteristic curves (AUCs), and background rejection $1/\epsilon_{\rm B}$ at a certain level of signal efficiency at 50\% and 30\%. The uncertainties correspond to the standard deviation in ten trainings. Several findings can be extracted from the table, with some explanations.

\begin{table*}[tb]
\setlength{\tabcolsep}{10pt}
\renewcommand{\arraystretch}{1.1}
\caption{Performance of the baseline network and the one supplemented by the pairwise patch structure with different variable designs. The baseline network is chosen from ParticleNet and \lorentznetbase. The model is trained on $60\,000$ jets from the training data of the top tagging dataset and evaluated on the full test data. The uncertainty is calculated from the standard deviation over ten trainings. For each metric, the best-performing networks from both ParticleNet and \lorentznetbase variants are highlighted in bold text.}
\label{tab:pair-perf}
\begin{center}
\begin{tabular}{ll|cccc}
\hline\hline
Base model  & Variation & Accuracy   & AUC   & \makecell[c]{$1/\epsilon_{\rm B}$\\$(\epsilon_{\rm S}=50\%)$}   & \makecell[c]{$1/\epsilon_{\rm B}$\\$(\epsilon_{\rm S}=30\%)$}     \\
\hline
\multirow{5}{*}{ParticleNet} &  ---     & $0.9310(3)$   & $0.9810(2)$  & $198\pm 7$   & $640\pm 29$   \\
&  +pairwise: $m_{ij}$        & $0.9334(8)$   & $0.9820(4)$  & $222\pm 13$   & $722\pm 52$   \\
&  +pairwise: $\Delta R_{ij}$  & $0.9334(6)$   & $\bf 0.9823(3)$  & $\bf 231\pm 10$   & $\bf 752\pm 43$   \\
&  +pairwise: $\Delta R_{ij}(p_{{\rm T},i} + p_{{\rm T},j})$  & $\bf 0.9337(3)$   & $0.9821(1)$  & $223\pm 6$   & $741\pm 36$   \\
&  +pairwise: $E_{ij}$     & $0.9303(5)$   & $0.9807(2)$  & $200\pm 6$   & $651\pm 23$   \\
\hline
\multirow{5}{*}{LorentzNet\textsubscript{base}} &  ---     & $0.9276(12)$   & $0.9789(7)$  & $172\pm 13$   & $581\pm 53$   \\
&  +pairwise: $m_{ij}$      & $\bf 0.9347(4)$   & $\bf 0.9829(2)$  & $\bf 260\pm 6$   & $\bf 931\pm 50$   \\
&  +pairwise: $\Delta R_{ij}$     & $0.9328(4)$   & $0.9819(3)$  & $232\pm 10$   & $807\pm 35$   \\
&  +pairwise: $\Delta R_{ij}(p_{{\rm T},i} + p_{{\rm T},j})$     & $0.9342(4)$   & $0.9826(2)$  & $251\pm 6$   & $919\pm 34$   \\
&  +pairwise: $E_{ij}$       & $0.9243(37)$   & $0.9767(23)$  & $144\pm 29$   & $485\pm 108$   \\
\hline\hline
\end{tabular}
\end{center}
\vskip 0.2in
\end{table*}

\begin{itemize}
\item For both ParticleNet and \lorentznetbase experiments, the network incorporating variables $m_{ij}$, $\Delta R_{ij}$, and $\Delta R_{ij}(p_{{\rm T},i} + p_{{\rm T},j})$ performs better than without using the pairwise features, and with injecting $E_{ij}$ with no dedicated symmetric design.
\item Comparing the three scenarios when cooperating with $m_{ij}$, $\Delta R_{ij}$, and $\Delta R_{ij}(p_{{\rm T},i} + p_{{\rm T},j})$, the ParticleNet experiment is more saturated in performance. On the other hand, \lorentznetbase cooperating with $m_{ij}$ and $\Delta R_{ij}(p_{{\rm T},i} + p_{{\rm T},j})$ are found to be more performant then $\Delta R_{ij}$. The latter finding matches to some degree with the fact that these two variables respect more underlying subsymmetries, showcased in Table~\ref{tab:pairwise-vars}.
\end{itemize}

To further reveal the relations between performance differences with the role in subsymmetries preservation, we evaluate the drop in performance when the test sample is processed by a given type of Lorentz transformation.

We study the case of \yz rotation, \xt boost, and $z$-tilt subsymmetries. Given the limitation that these transformations are done when the jet is directed to the $x$ axis, for a given jet, the overall transformation is described as $(\Lambda_0)^{-1}(\Lambda_{\rm target})(\Lambda_0)$, where $\Lambda_0$ is a successive \zt boost and \xy rotation to ensure the jet points to $x$ direction, and $\Lambda_{\rm target}$ is the target transformation (i.e., \yz rotation, \xt boost, or $z$ tilt). The performance in AUC under various test datasets transformation is shown in Fig~\ref{fig:lorentz-attk-pair}.

\begin{figure*}[tb]
\begin{center}
\includegraphics[width=0.32\textwidth]{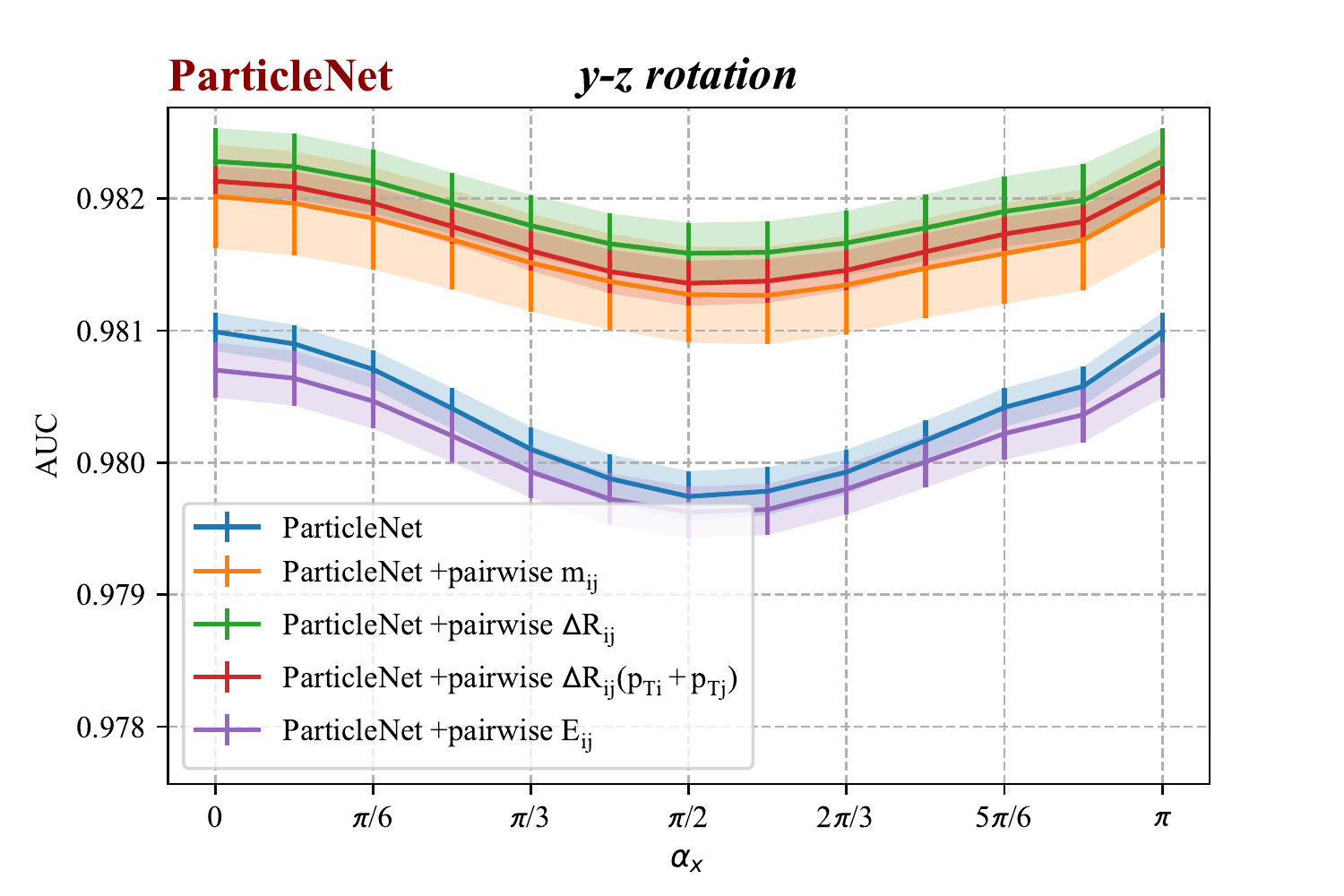}
\includegraphics[width=0.32\textwidth]{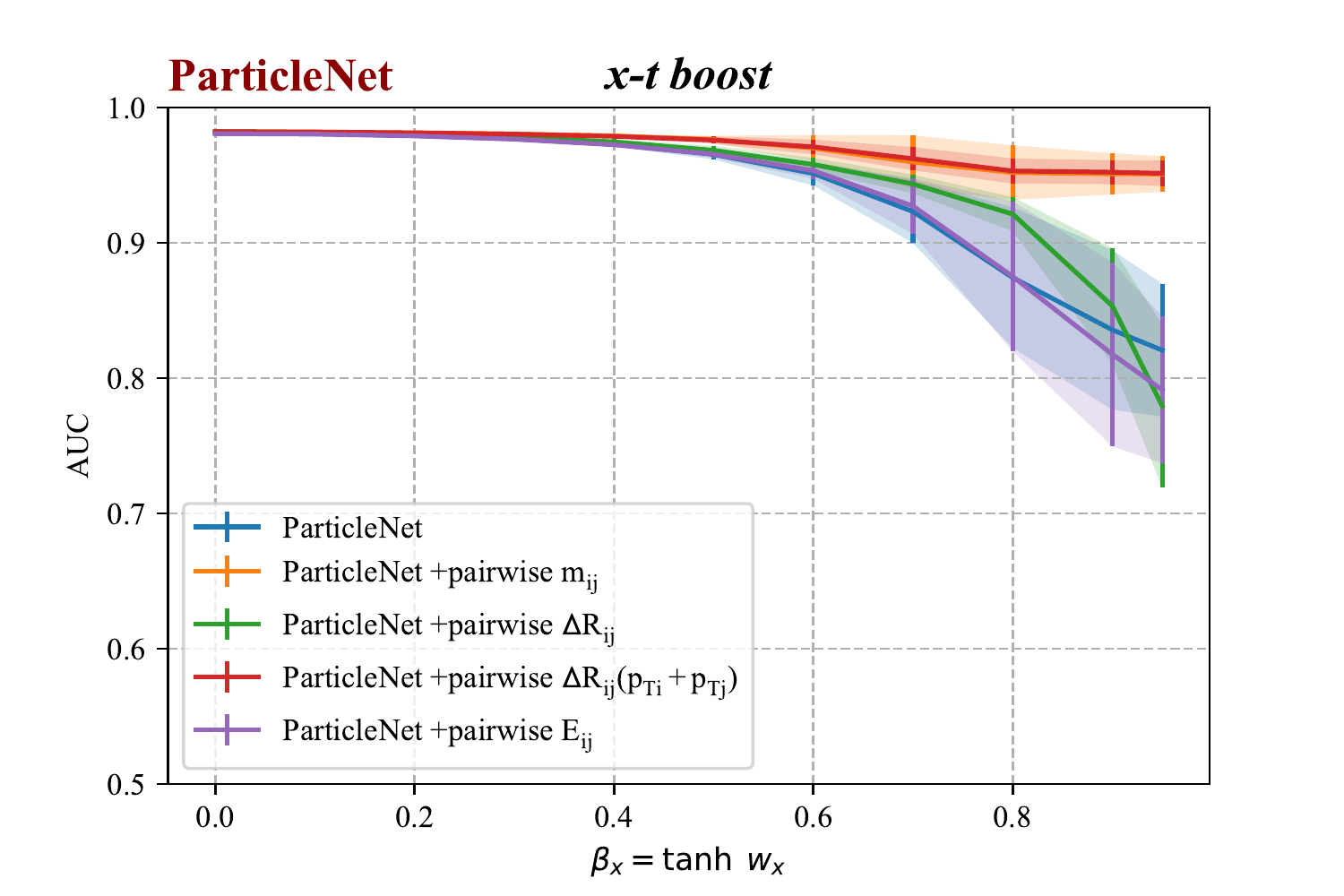}
\includegraphics[width=0.32\textwidth]{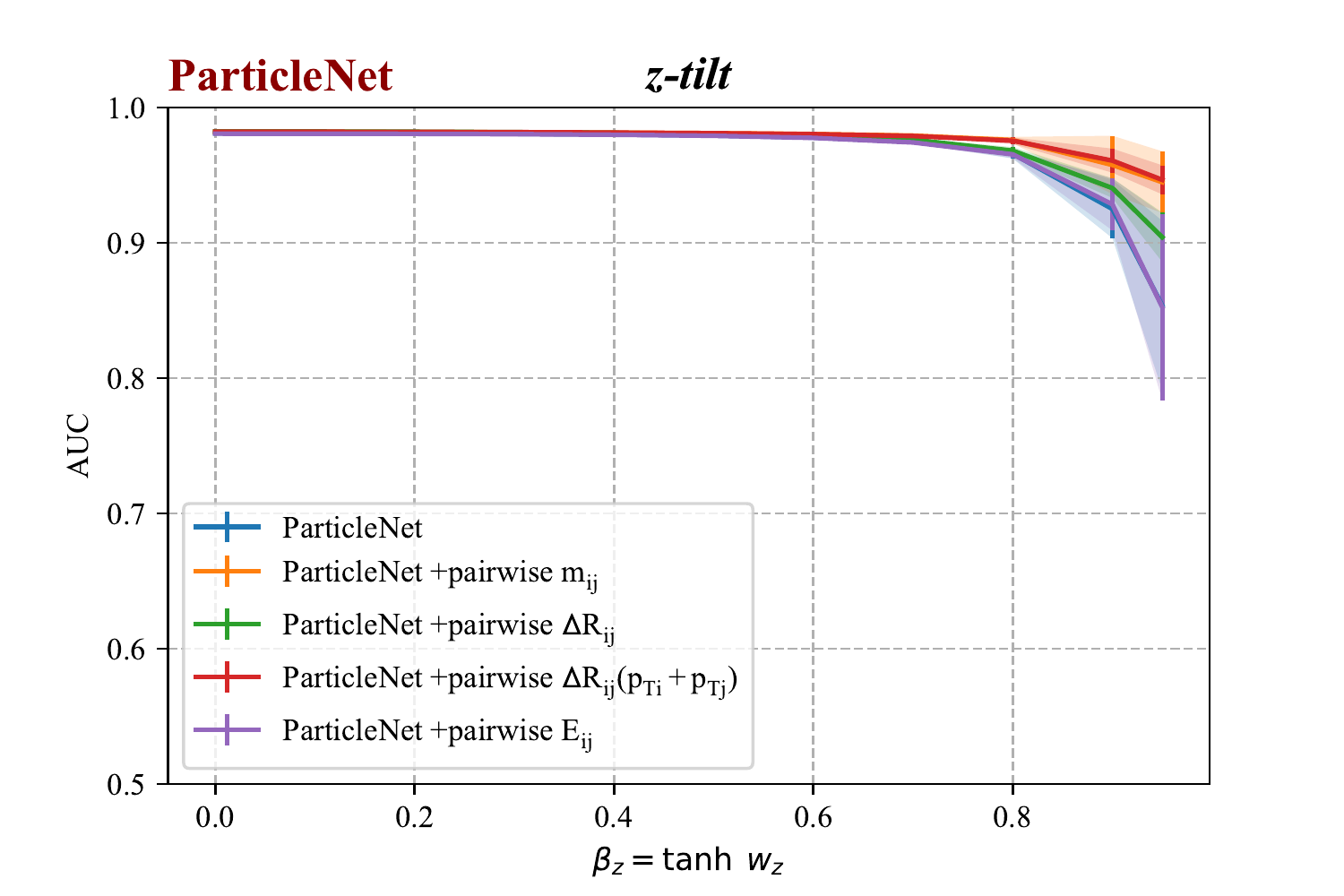} \\
\includegraphics[width=0.32\textwidth]{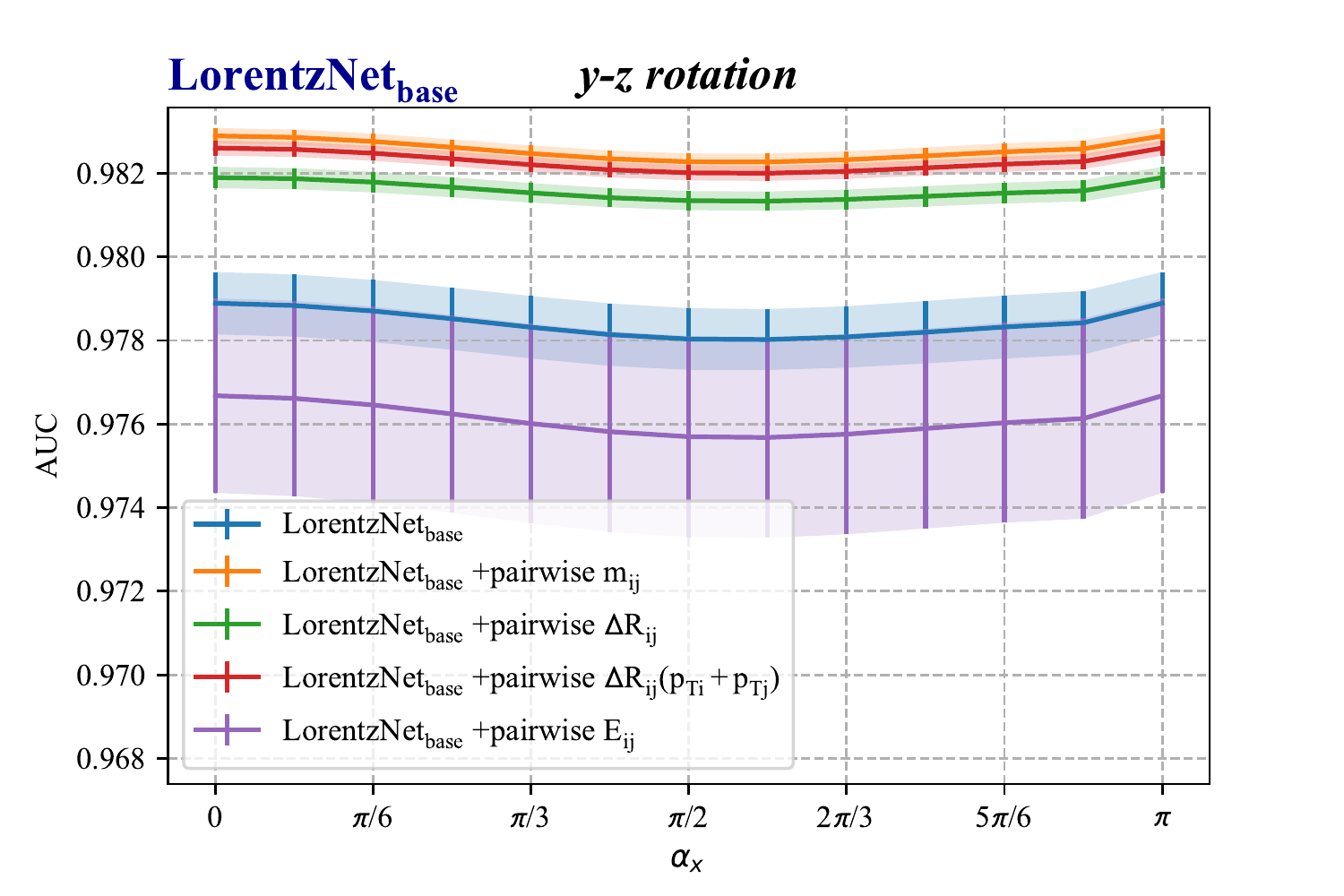}
\includegraphics[width=0.32\textwidth]{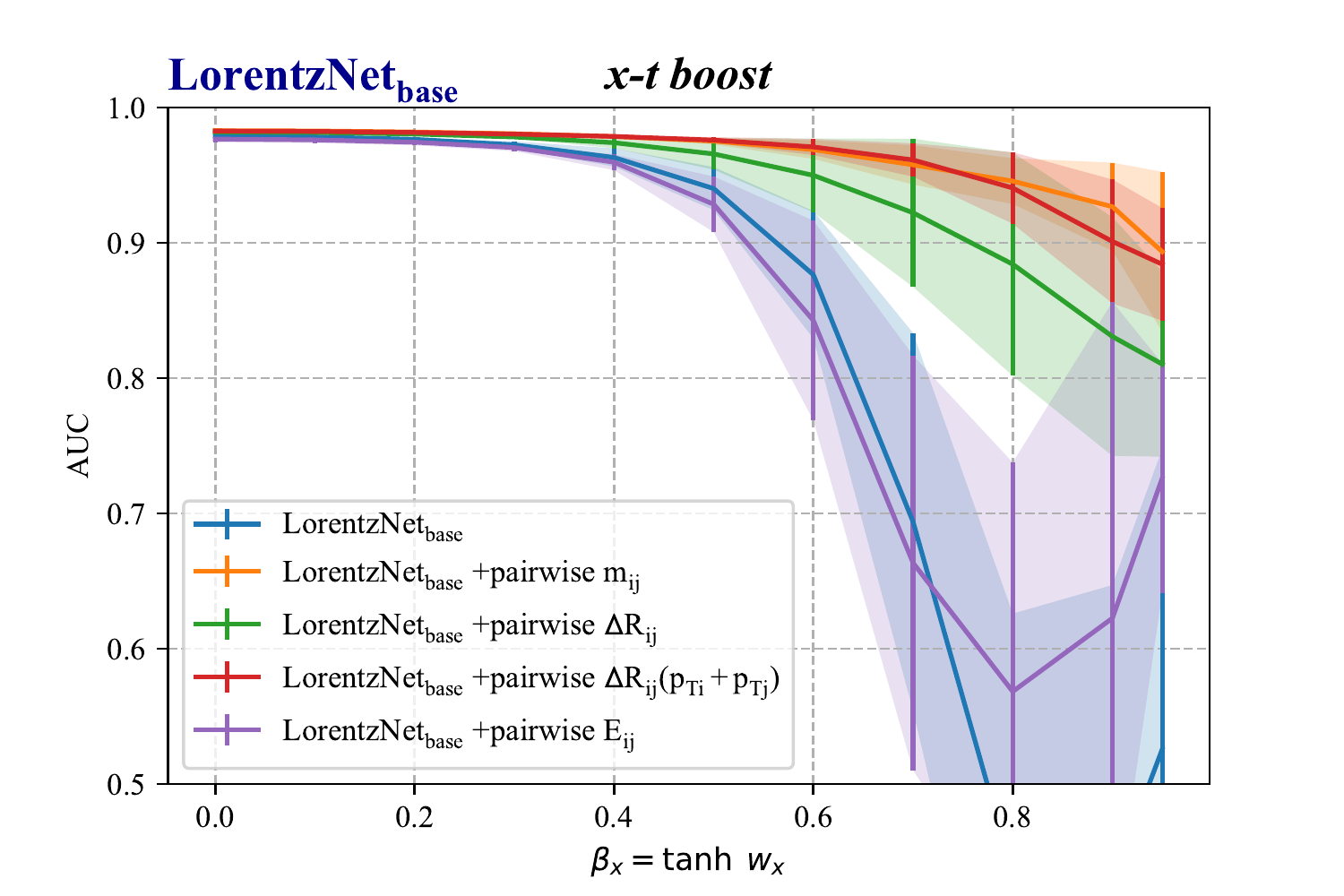}
\includegraphics[width=0.32\textwidth]{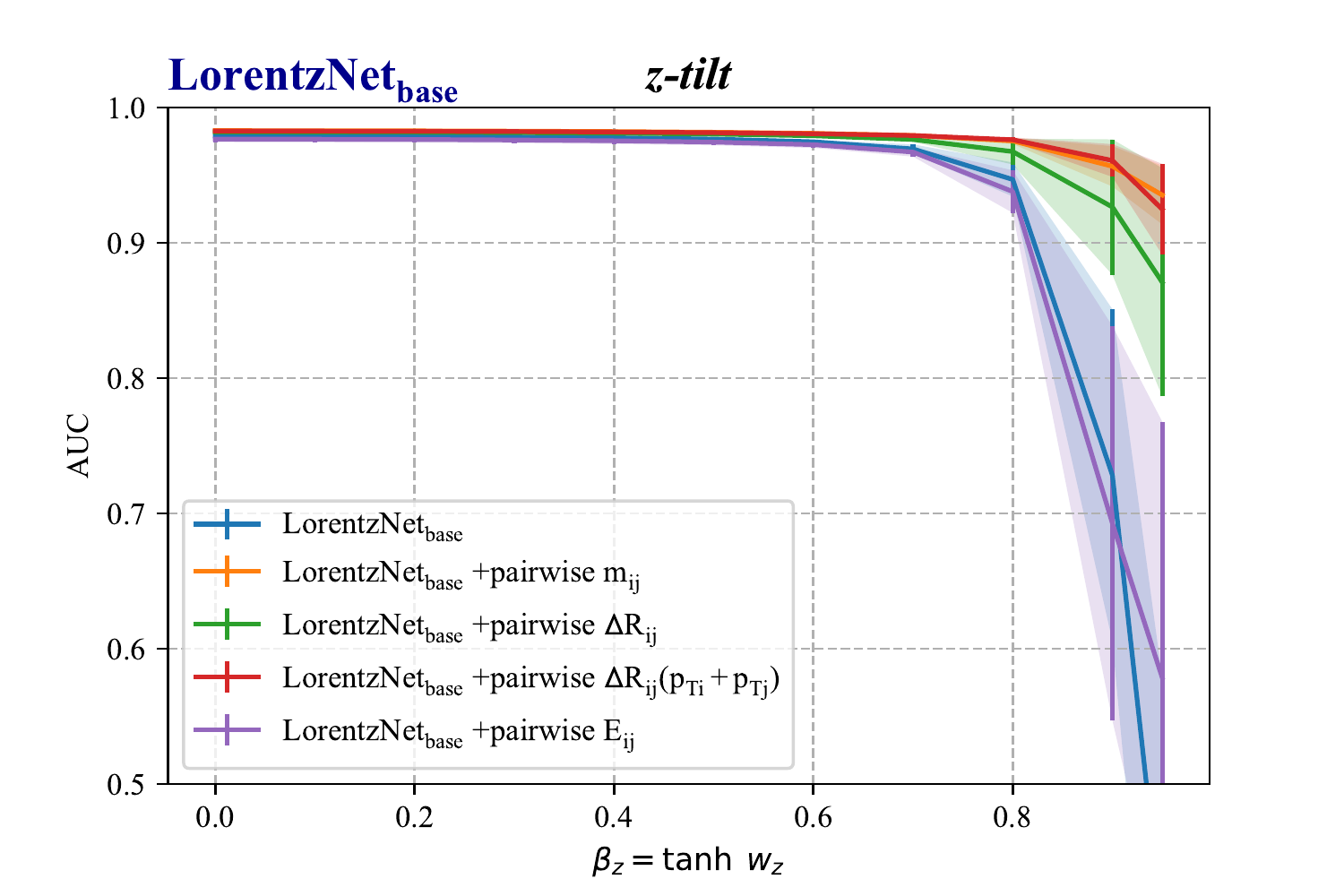} \\
\caption{Network performance in terms of AUC evaluated under various types of Lorentz transformation applied to the test dataset. From left to right: the input jet undergoes a \yz rotation with angle $\alpha_x$,  an \xt boost with rapidity $w_x$, and a $z$-tilt transformation (i.e., a \zt boost with rapidity $w_z$ followed by a \xz rotation to redirect the jet to the $x$ axis). The curves in the plots show different options of pairwise features added to the baseline network. The baseline is chosen as ParticleNet (top) or \lorentznetbase (bottom). The error bar shows the standard deviation over ten trainings.}
\vskip 0.2in
\label{fig:lorentz-attk-pair}
\end{center}
\end{figure*}

We summarize new findings from the plots as follows.
\begin{itemize}
\item The first obvious finding is that, when the added patch structure is invariant regarding a certain symmetry, the whole network tends to be more resilient to the transformations on that symmetry. Specifically, the network adding the ``mass'' recipe becomes more robust for all transformation scenarios; the network adding the patch incorporating $\Delta R_{ij}$ has improvement on adaptability for \yz rotations; and interestingly, adding $\Delta R_{ij}(p_{{\rm T},i} + p_{{\rm T},j})$ improves adaptability for both \yz rotations and \xt boost. This suggests that the added patch plays a significant role during network training, as its invariant properties are, to some extent, imparted to the entire network.
\item Networks adding the symmetry-preserving structure show smaller spreads on the metric over ten trainings. Especially, the ablation case that adds $E_{ij}$ shows unstable training results. This can be attributed to the patch structure disrupting the more fundamental symmetry associated with \zt boost. Overall, they indicate that networks with higher levels of symmetry preservation exhibit a stronger generalization ability.
\end{itemize}

The above observations can be interpreted by recognizing that preserving Lorentz symmetry acts as a special ``inductive bias'' for the jet tagging tasks. Generally, the inductive bias works in the principle as follows. By introducing such a patch network structure that remains invariant under any Lorentz transformation, we effectively provide a hint to our network that the input jet property (e.g., its truth label) typically does not change when the input jet undergoes any Lorentz transformation. An analog example is the benefit of the CNN architecture in the vision domain (such as in the image classification task), where the specialized design of CNNs can provide a hint that the input image properties are generally unaffected by shifts of some elements within the image.

From this viewpoint, the two observations can be explained as follows. Our first observation, i.e., greater resilience to the transformations with the introduction of more symmetry levels, can be seen as evidence that the network benefits from this inductive bias.
This benefit arises as the network relies on the symmetry-preserving property of the patch structure and eventually propagates it throughout its entire structure.
The second finding, on the other hand, can be understood by acknowledging that integrating an inductive bias into the network effectively serves as a method of augmenting input data.

As a further validation that the symmetry-preserving property serves as an inductive bias, Fig.~\ref{fig:datasize-pair} shows the original top tagging performance in terms of the AUC when the network is trained on different sample sizes, ranging from 6000, $12\,000$, and $60\,000$ jets.
Clearly, the modified network that preserves more levels of symmetries performs better in the low-data scheme.
As incorporating inductive bias generally helps networks perform better on small samples due to its effective data augmentation, this is again in line with our observation.
Figure~\ref{fig:datasize-pair} also shows that, when the data size rises, the top tagging performance on this dataset tends to converge. We believe that it is caused by the performance saturation for this benchmark task when training with larger data.
To verify that the inductive bias is not confined to specific training sizes but is a general property to enhance the network performance, we conduct an additional experiment using the JetClass dataset. This experiment employs a more intricate ten-class classification to assess the impact of inductive bias on a $100\times 10^{6}$ dataset. Figure~\ref{fig:datasize-pair-jetclass} shows the jet tagging performance on JetClass, measured in terms of the multiclass classification AUC, as the training datasets range from $60\,000$ to $100\times 10^{6}$. The study only utilizes ParticleNet baseline and its two variants: one adding the patch structure with $m_{ij}$ as input to preserve full Lorentz symmetries and the other an ablation case with $E_{ij}$, which disrupts an existing symmetry related to the \zt boost. The results consistently support our conclusion, even with large datasets, while signs of performance saturation convergence appear less pronounced.

\begin{figure*}[tb]
\begin{center}
\includegraphics[width=0.40\textwidth]{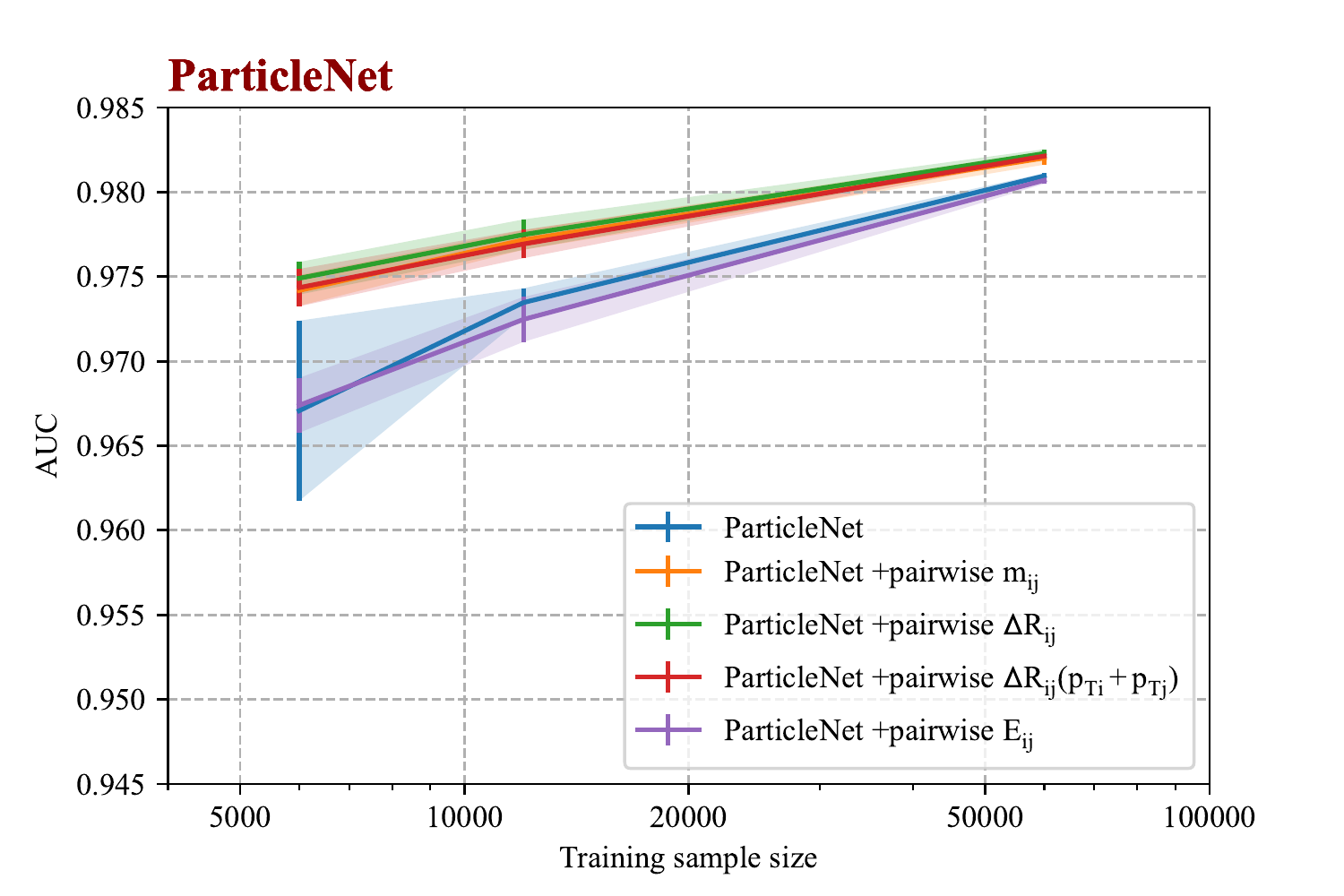}
\includegraphics[width=0.40\textwidth]{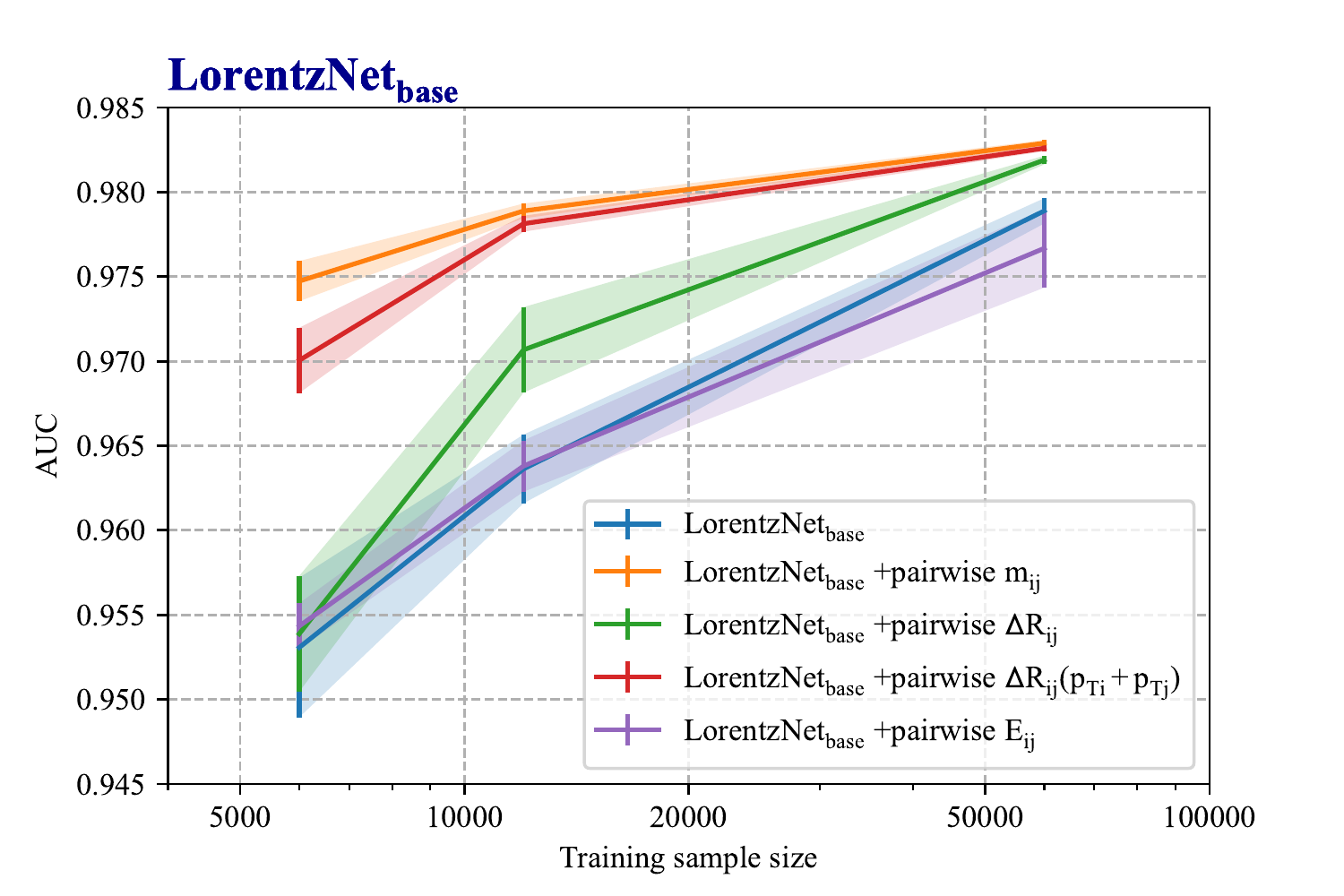} \\
\caption{Network performance in the top tagging task in terms of AUC versus the training size, selected among \{6000, $20\,000$, $60\,000$\}. The curves in the plots show different options of pairwise features added to the baseline network. The baseline is chosen as ParticleNet (left) or \lorentznetbase (right). The error bar shows the standard deviation over ten trainings.}
\label{fig:datasize-pair}
\end{center}
\end{figure*}

\begin{figure}[tb]
\begin{center}
\includegraphics[width=0.40\textwidth]{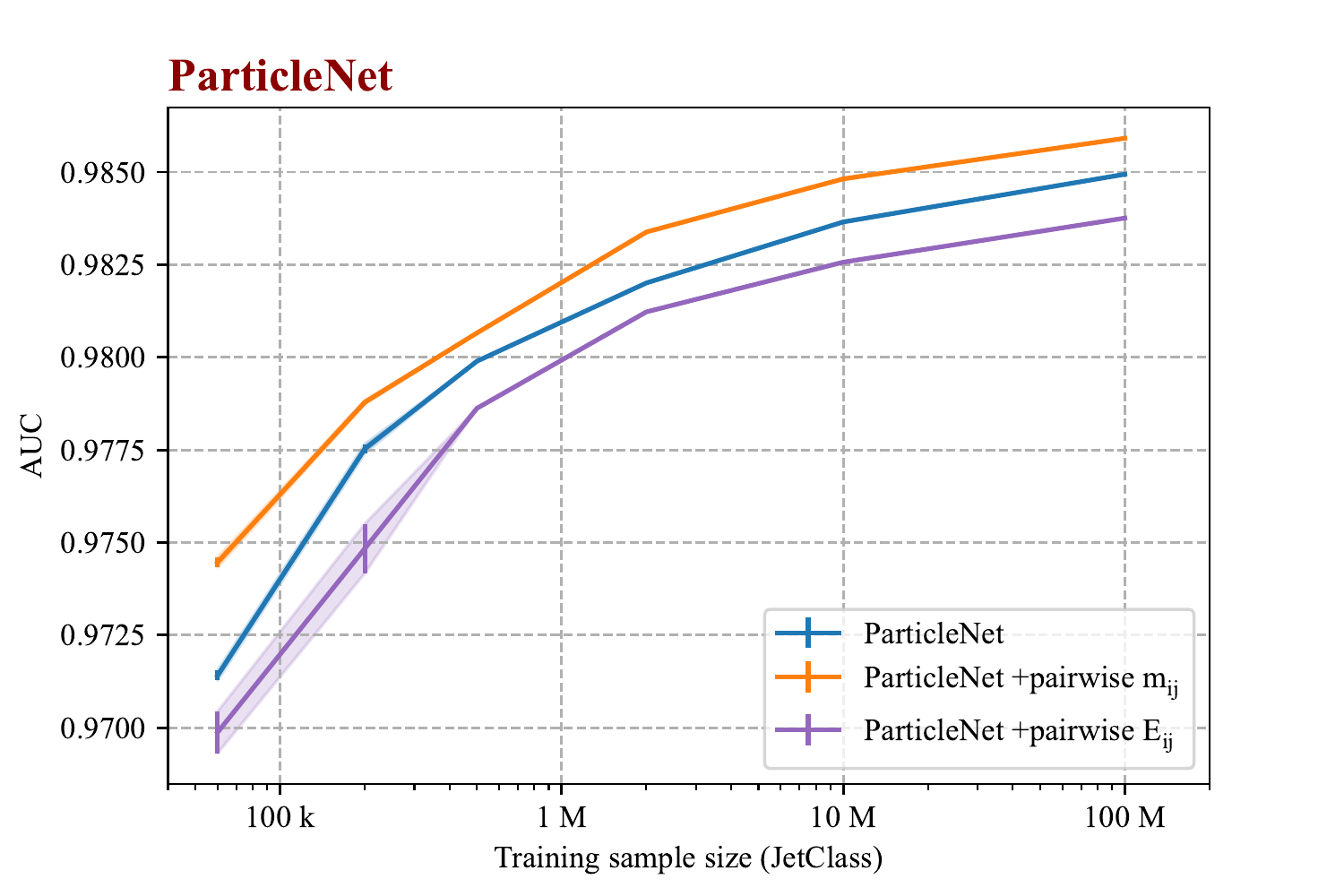} \\
\caption{Network performance in the JetClass classification task in terms of AUC versus the training size, selected among \{$60\,000$, $200\,000$, $500\,000$, $2\times 10^{6}$, $10\times 10^{6}$, $100\times 10^{6}$\}. The curves in the plots show different options of pairwise features added to the ParticleNet network as the baseline. The training iterations are conducted 10 times for the $60\,000$ case and 5 times for the $200\,000$ case, with the error bar showing the standard deviation over multiple trainings. Training with larger datasets is performed only once.}
\label{fig:datasize-pair-jetclass}
\end{center}
\end{figure}

Before ending this section, we offer a final remark on the pairwise mass, $m_{ij}$.
The above results also provide a new angle to explain the benefits brought by pairwise mass. As a matter of fact, the mass variable is generally considered an important nonkinematics feature that sculpts the dynamics properties of the physics system. This is generally used by experimentalists to explain that this variable can play a crucial role in the multivariate analysis, which other event kinematics features cannot compete with. We can, however, modify from a pure kinematics feature---separating angle $\Delta R$---from the hints of symmetries to achieve similar performance with the mass feature. This brings a new angle to interpret the role of pairwise mass participating in the network that boosts the performance. However, we need to point out that there are actually mathematical relations between $m_{ij}^2$ and $\Delta R_{ij}$. In the relativistic limit and considering $y_{y,z}\sim o(1)$ for the particles $i$, $j$, we can derive (with proof in Appendix~\ref{sec:proof})
\begin{equation}
m_{ij}^2 = |\mathbf{p}_i||\mathbf{p}_j|(1-\cos\,\theta_{(\mathbf{p}_i,\,\mathbf{p}_j)}) \approx \frac{1}{2}\Delta R_{ij}^2 p_{{\rm T},i}\,p_{{\rm T},j},
\end{equation}
which means the pairwise mass can be equivalently considered as another form of \pt-weighted angular separation feature between particles. Nevertheless, the fact that a pure mathematically constructed variable, $\Delta R_{ij}(p_{{\rm T},i} + p_{{\rm T},j})$, is able to rank a high performance as well and manifest expected behavior on imposing various types of transformations on the test dataset is sufficient to illustrate the role symmetry plays behind the network's mechanism.

\subsection{Incorporating nodewise features}

Our above study has used additional implementation of pairwise features to illustrate the role Lorentz-symmetric design plays in network training. However, the pairwise features have limited usage as they are generally applicable to GNN or attention-based baseline models only. Therefore, we consider further extending its application scheme and hope to design a more generalized patch. The new patch structure is based on additional nodewise features and can be applied to all mainstream networks that rely on the point-cloud (set) representation of the input data.

\subsubsection{Variables}

The node features are designed from the same spirit to incorporate mass variables, but carried in a nodewise manner instead of pairwise. For each node $i$, we define a group of friend nodes $G_i$, the choice of which is invariant under Lorentz transformations. We calculate their invariant mass
\begin{equation}\label{eq:mass-linear-comb}
m_{G_i} = \Big[\Big(\sum_{j\in G_i} p_j\Big)^\mu\Big(\sum_{j\in G_i} p_{j}\Big)_\mu \Big]^{\frac{1}{2}} = \Big[2\sum_{j,k\in G_i}^{j<k}p_j^\mu p_{k,\mu}\Big]^{\frac{1}{2}}.
\end{equation}
Therefore, it is essentially the predetermined linear combination of all Lorentz scalars $p_i^\mu p_{j,\mu}$. In the ablation study, the nodewise variable $E_{G_i} = \sum_{j\in G_i} E_j$ is also considered as an option.

We deliver studies in the determination of the $G_i$. We find that
\begin{equation}
G_i = \{ j \Big| p_{i}^\mu p_{j,\mu}\text{ is among the $k$ largest values for all $j$}\}
\end{equation}
is a Lorentz-invariant choice and makes the network more performant. Here, $k$ is a predetermined variable. We choose the value of $k$ as \{4, 8, 16, 32\}, hence creating the nodewise features with a dimension of 4.

\subsubsection{Patch structure}

The injection of new nodewise features to the baseline network is created as a rather generic design. Therefore, we use PFN, ParticleNet, and \lorentznetbase as our baseline for experiments and implement the same patch structure to all three networks. The patch is illustrated in Fig.~\ref{fig:patch} (b). First, $N$ nodewise features are calculated at the beginning stage, and embedded from the initial dimension 4 to the fixed feature dimension 64 by an elementwise MLP, via two hidden layers of feature dimension 64. The embedded nodewise feature is denoted by $\mathbf{u}_i$. We note that all mainstream networks viewing jets as a point cloud (set) are composed of a stack of some unit block to update the nodewise features of the particles.
For PFN, the unit block is $\Phi(x)$ according to the notation from Ref.~\cite{Komiske:2018cqr}. This represents a feed-forward network designed to individually update particle features.
For ParticleNet, it is the EdgeConv operation~\cite{Qu:2019gqs}; for LorentzNet, it is the Lorentz Group Equivariant Block~\cite{Gong:2022lye}. We need to incorporate our additional node feature $\mathbf{u}_i$ into the existing structure block by block. In the data processing flow, we update the node feature $\mathbf{x}_i$ which will be fed into the unit block by $\mathbf{u}_i$, after a dimension-matching linear layer. This can be expressed by
\begin{equation}
    \mathbf{x}'_{i} = \mathbf{x}_i + \text{\texttt{Linear}}(\mathbf{u}_i).
\end{equation}
We note the injection strategy is very similar to that of including pairwise features in ParticleNet, by comparing to Fig.~\ref{fig:patch} (a) and the injection formula described in Eq.~(\ref{eq:pair-conv2d}). Both methods employ an embedding of the additional features first and then inject them into the baseline network block by block. The difference is that our current features are on a per-node basis and are more generalized to be applied.

\subsubsection{Experiments}

We do the same experiments as detailed in Sec.\ref{sec:pair-exp} to study the effect when incorporating the additional node features via our generalized mechanism and understand its relation with symmetry preservation. Table~\ref{tab:node-perf} shows the performance of different schemes in terms of accuracy, the AUC, and background rejections. Figure~\ref{fig:lorentz-attk-node} shows the robustness study of network performance upon Lorentz transformations on the test dataset. Figure~\ref{fig:datasize-node} provides the performance trend when trained on various sample sizes. All uncertainties shown in the table and plots correspond to the standard deviation over ten trainings. The findings are, overall, similar to the pairwise case in Sec.\ref{sec:pair-exp}, and are summarized below.

\begin{table*}[tb]
\setlength{\tabcolsep}{10pt}
\renewcommand{\arraystretch}{1.1}
\caption{Performance of the baseline network and the one supplemented by the nodewise patch structure with different variable designs. The baseline network is chosen from PFN, ParticleNet, and \lorentznetbase. The model is trained on $60\,000$ jets from training data and evaluated on the full test data. The uncertainty is calculated from the standard deviation over ten trainings. For each metric, the best-performing networks from both ParticleNet and \lorentznetbase variants are highlighted in bold text.}
\label{tab:node-perf}
\begin{center}
\begin{tabular}{ll|cccc}
\hline\hline
Base model  & Variation & Accuracy   & AUC   & \makecell[c]{$1/\epsilon_{\rm B}$\\$(\epsilon_{\rm S}=50\%)$}   & \makecell[c]{$1/\epsilon_{\rm B}$\\$(\epsilon_{\rm S}=30\%)$}     \\
\hline
\multirow{3}{*}{PFN} &  ---    & $0.9104(12)$   & $0.9664(13)$  & $67\pm 5$   & $198\pm 21$   \\
    &  +nodewise: $m_{G_i}$   & $\bf 0.9281(4)$   & $\bf 0.9791(2)$  & $\bf 184\pm 5$   & $\bf 714\pm 50$   \\
    &  +nodewise: $E_{G_i}$   & $0.9207(4)$   & $0.9750(3)$  & $125\pm 3$   & $378\pm 19$   \\
\hline      
\multirow{3}{*}{ParticleNet} &  ---      & $0.9310(3)$   & $0.9810(2)$  & $198\pm 7$   & $640\pm 29$   \\
    &  +nodewise: $m_{G_i}$     & $\bf 0.9313(3)$   & $\bf 0.9812(1)$  & $\bf 222\pm 5$   & $\bf 800\pm 40$   \\
    &  +nodewise: $E_{G_i}$     & $0.9300(12)$   & $0.9802(6)$  & $183\pm 12$   & $572\pm 47$   \\
\hline
\multirow{3}{*}{LorentzNet\textsubscript{base}} &  ---      & $0.9276(12)$   & $0.9789(7)$  & $172\pm 13$   & $581\pm 53$   \\
    &  +nodewise: $m_{G_i}$      & $\bf 0.9306(3)$   & $\bf 0.9809(2)$  & $\bf 219\pm 3$   & $\bf 887\pm 36$   \\
    &  +nodewise: $E_{G_i}$      & $0.9272(3)$   & $0.9788(1)$  & $171\pm 2$   & $562\pm 16$   \\
\hline\hline
\end{tabular}
\end{center}
\end{table*}

\begin{figure*}[tb]
\begin{center}
\includegraphics[width=0.32\textwidth]{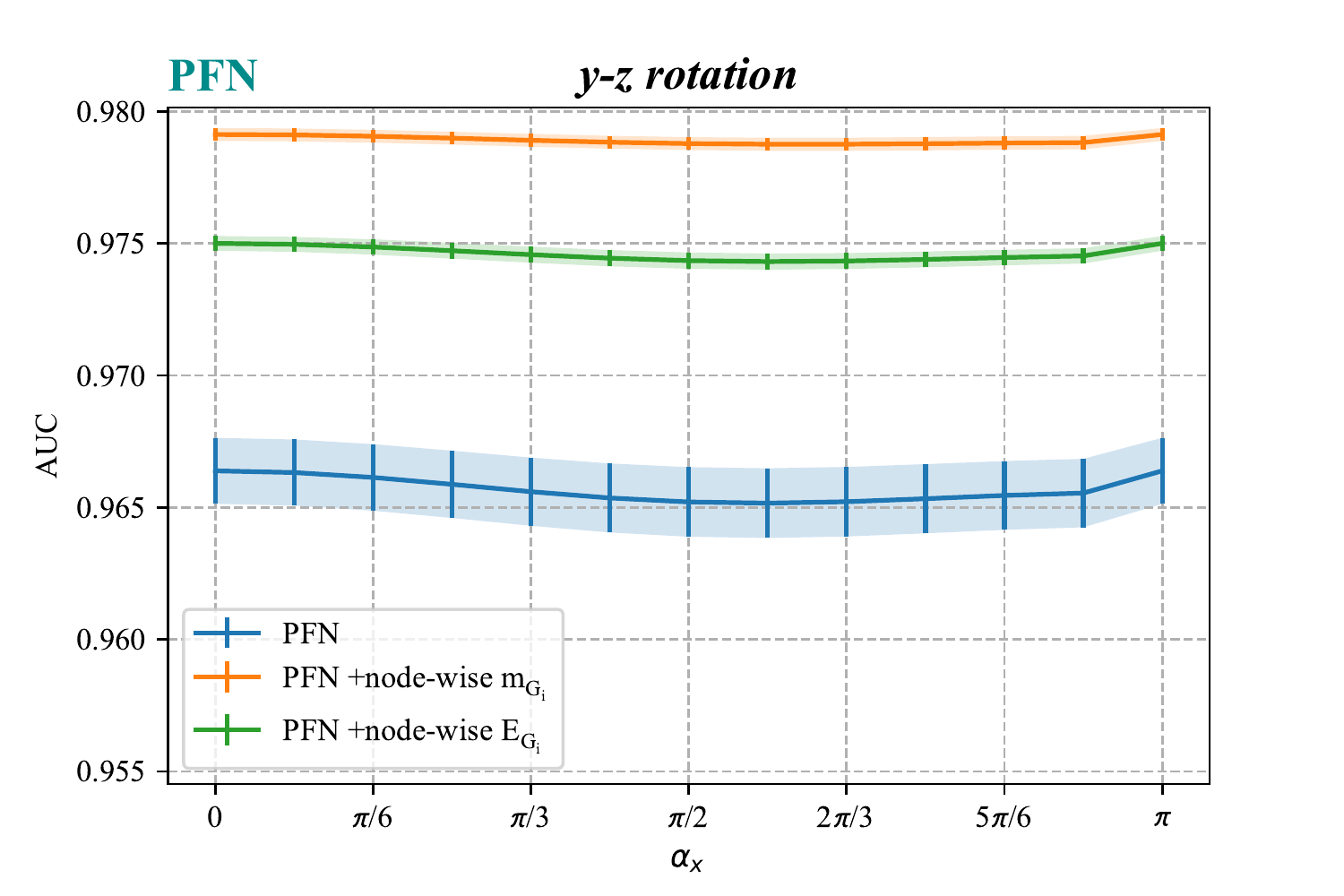}
\includegraphics[width=0.32\textwidth]{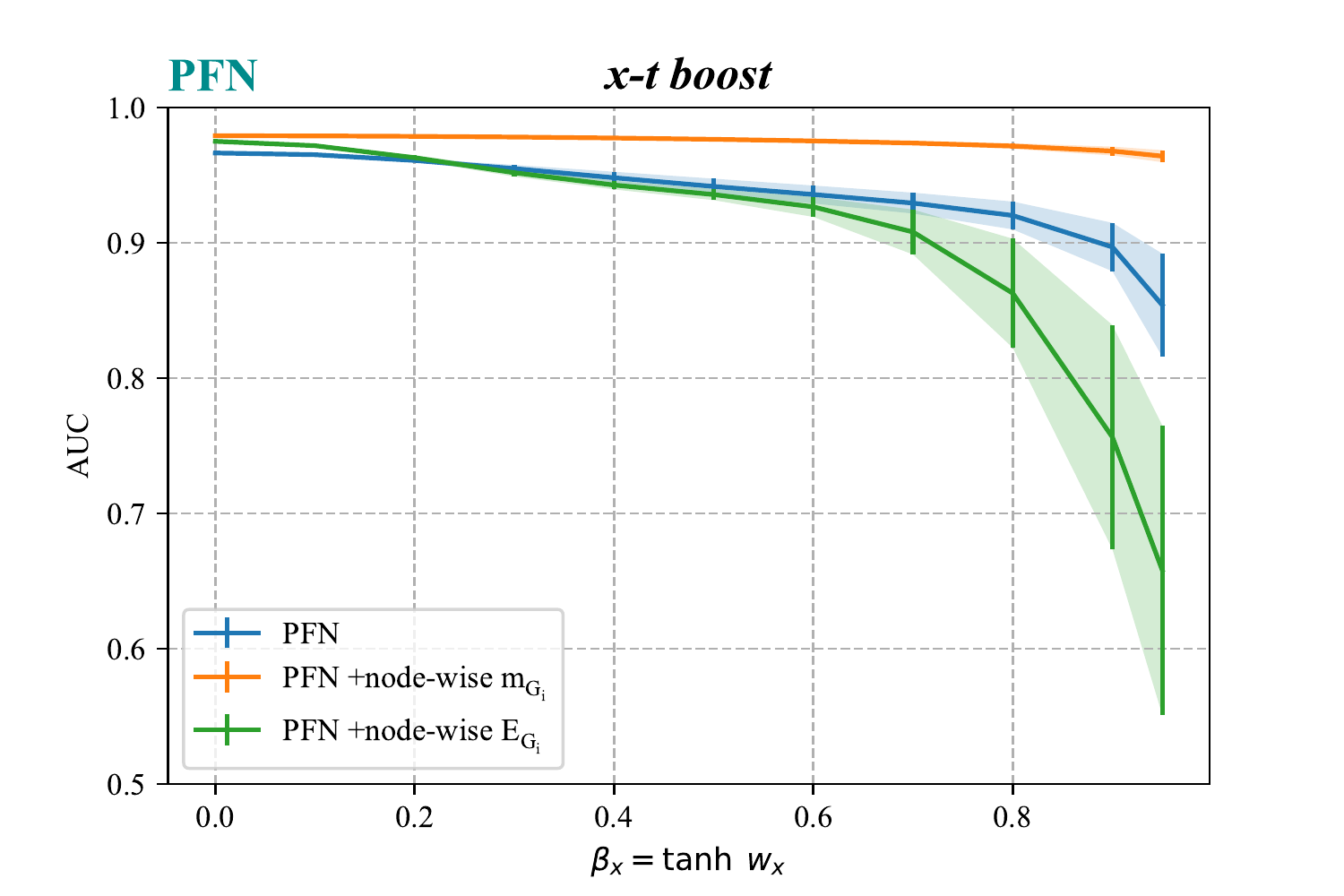}
\includegraphics[width=0.32\textwidth]{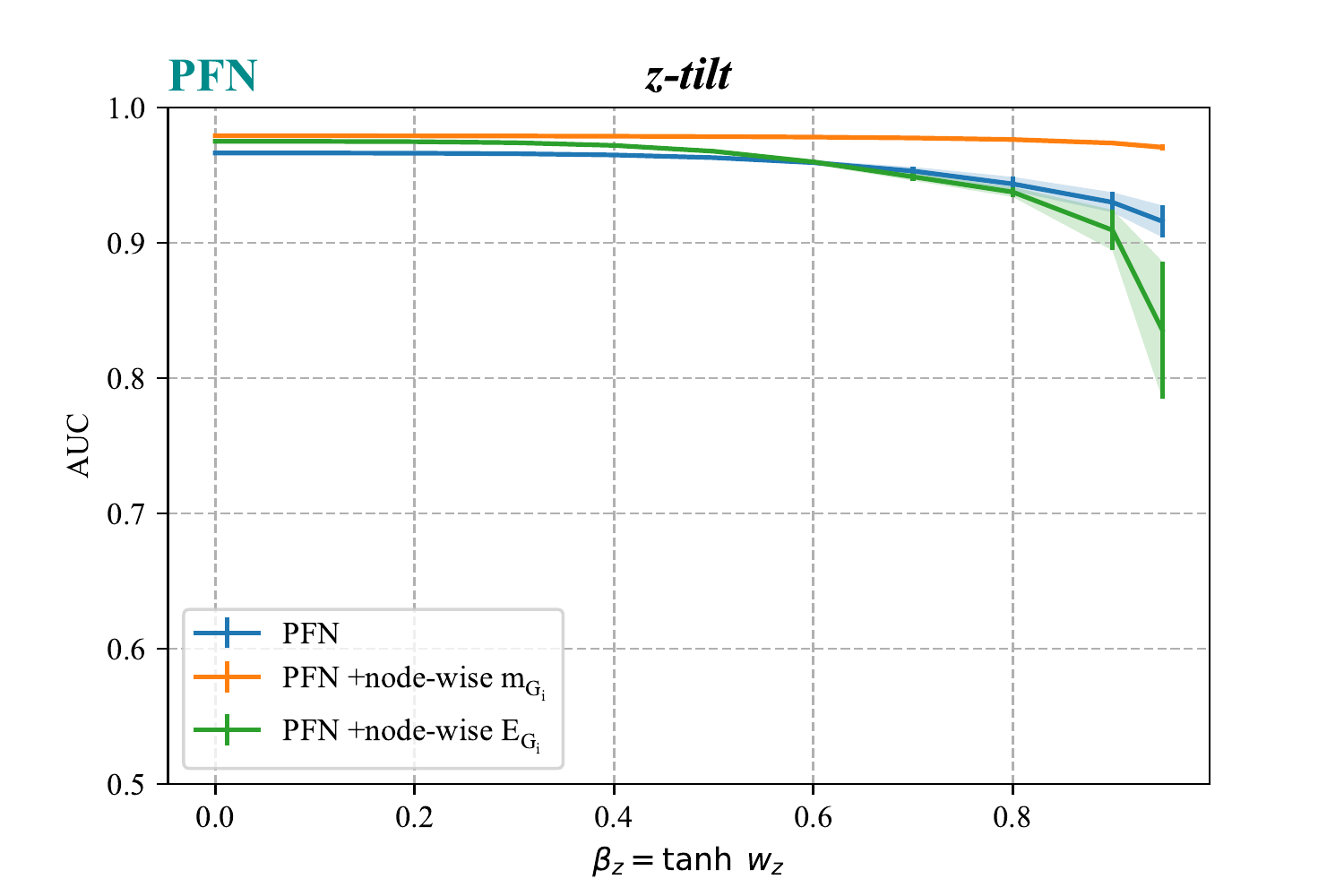} \\
\includegraphics[width=0.32\textwidth]{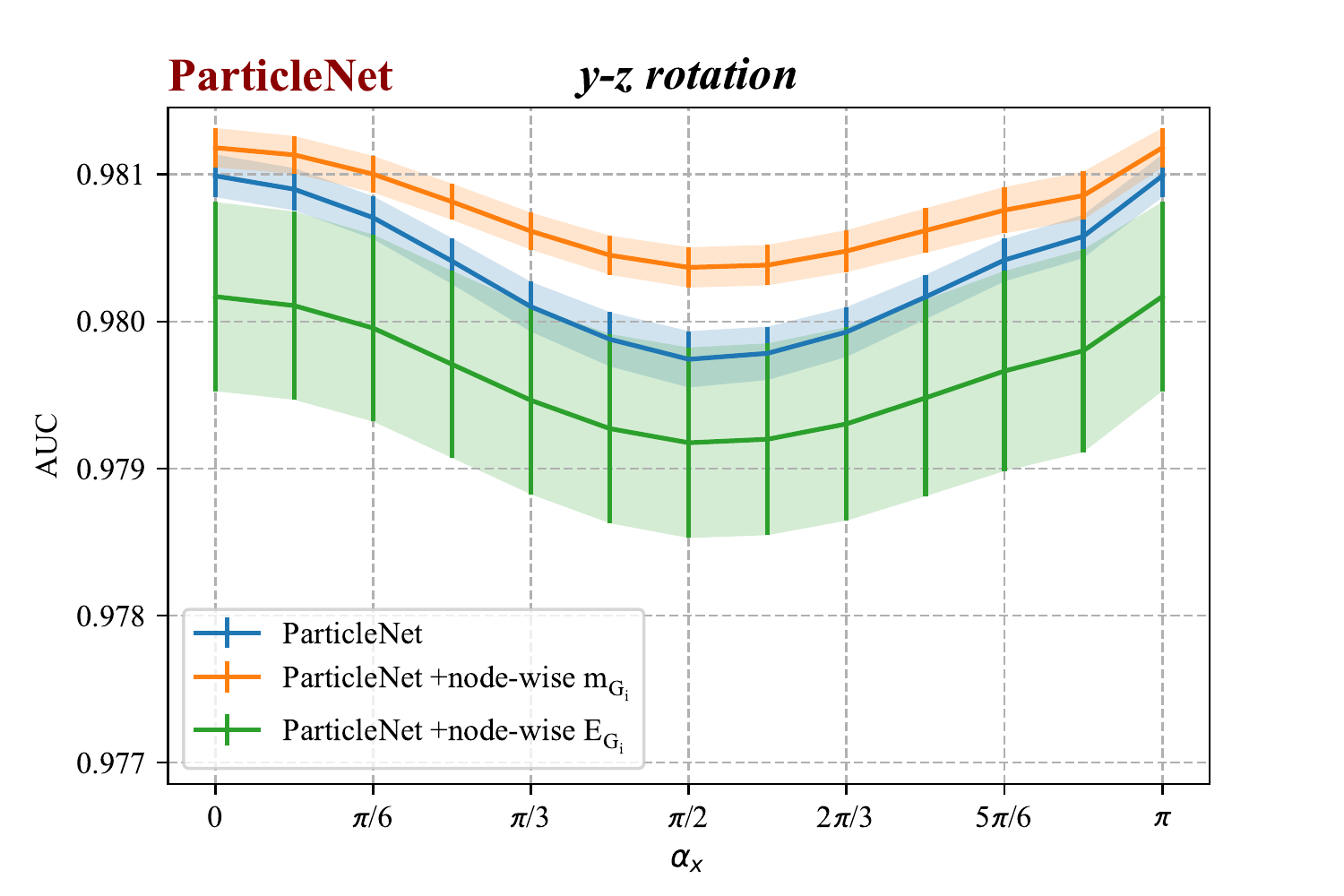}
\includegraphics[width=0.32\textwidth]{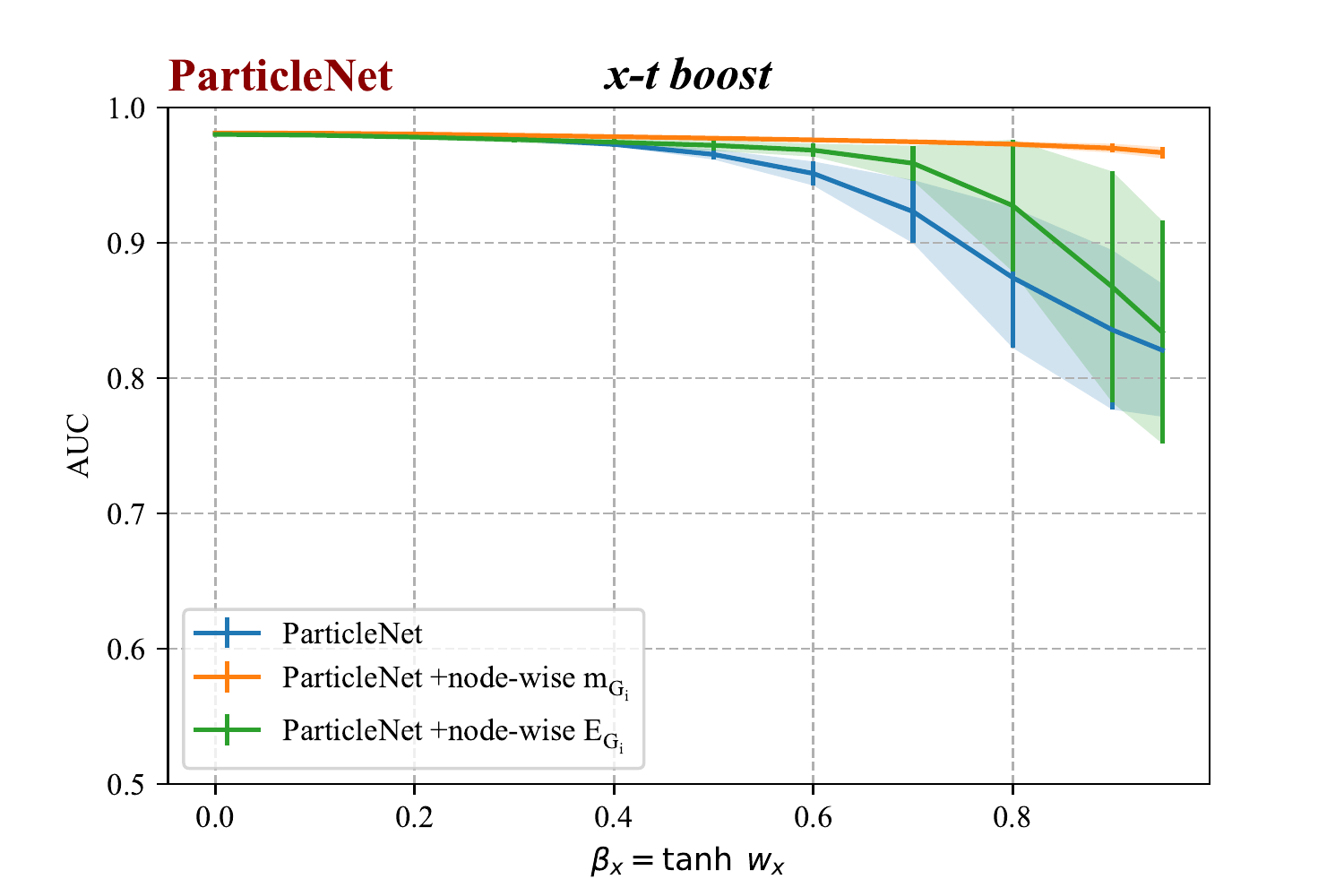}
\includegraphics[width=0.32\textwidth]{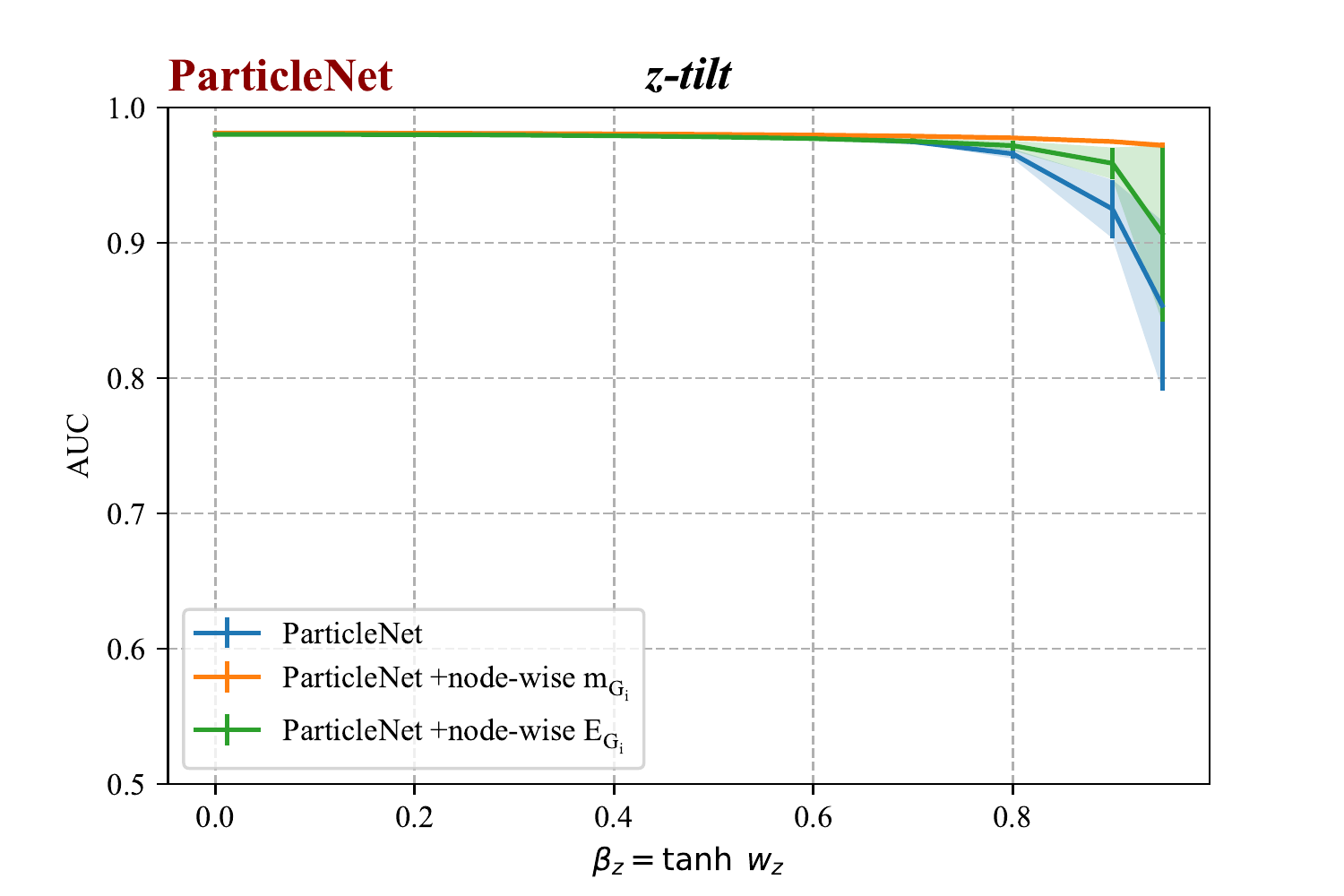} \\
\includegraphics[width=0.32\textwidth]{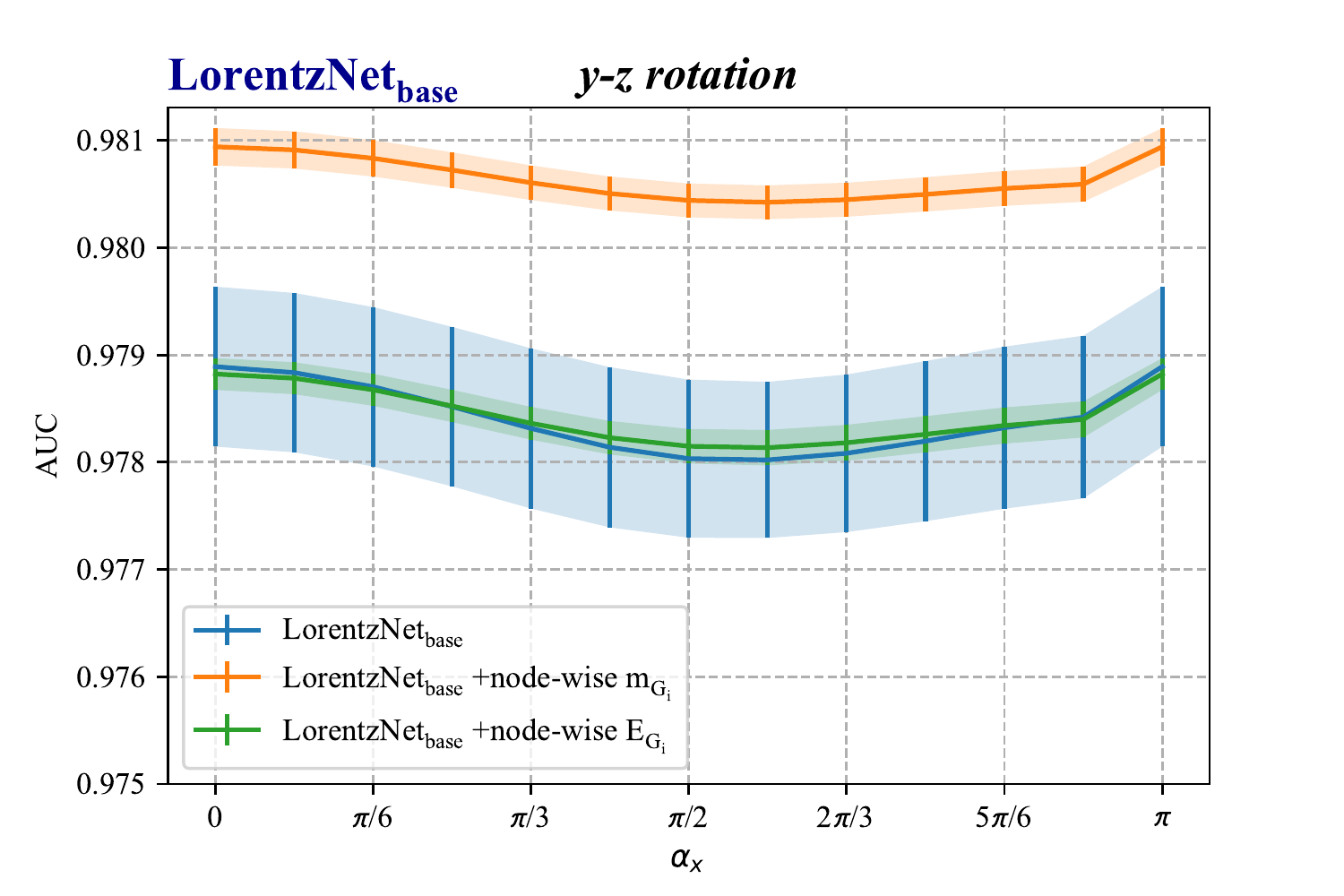}
\includegraphics[width=0.32\textwidth]{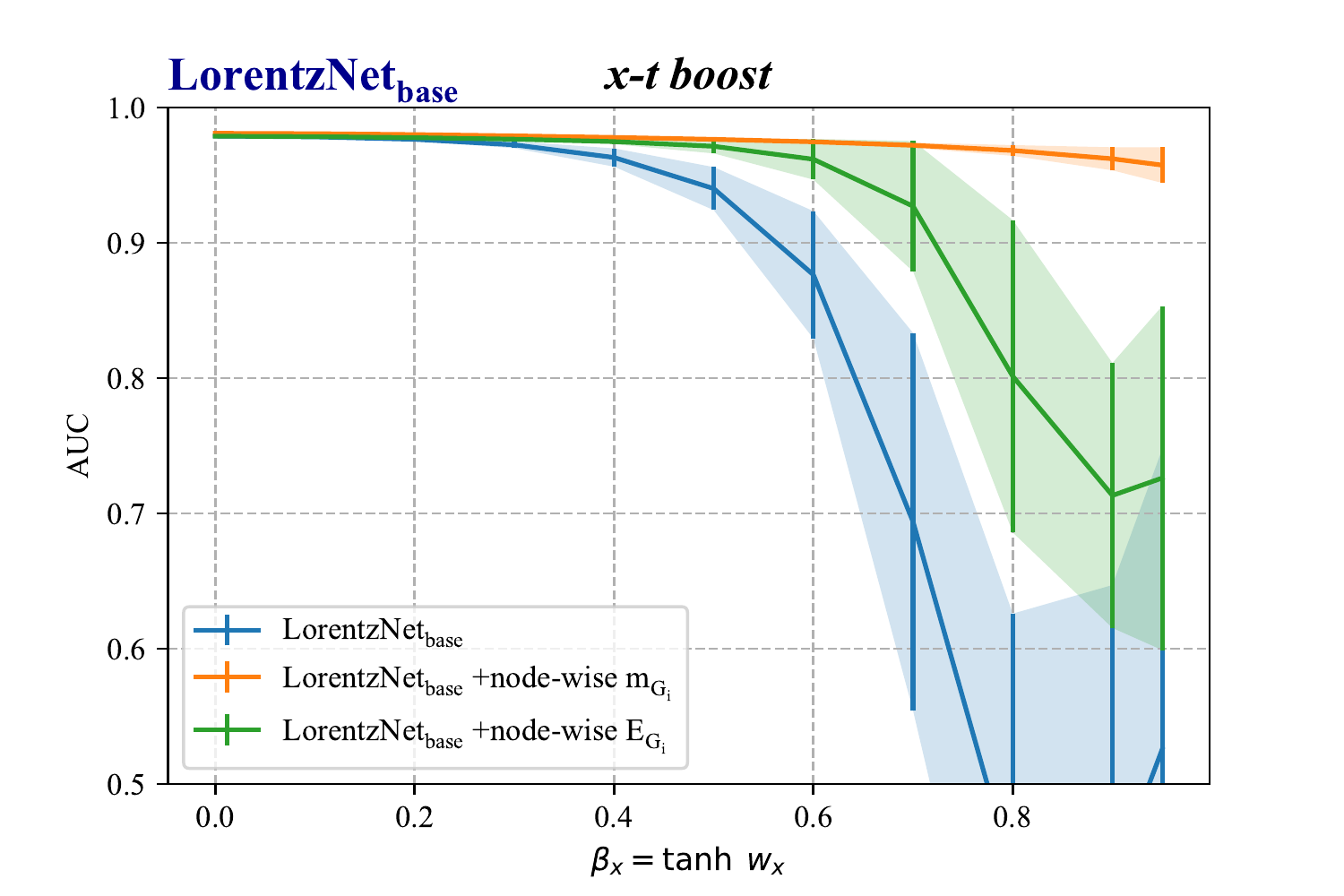}
\includegraphics[width=0.32\textwidth]{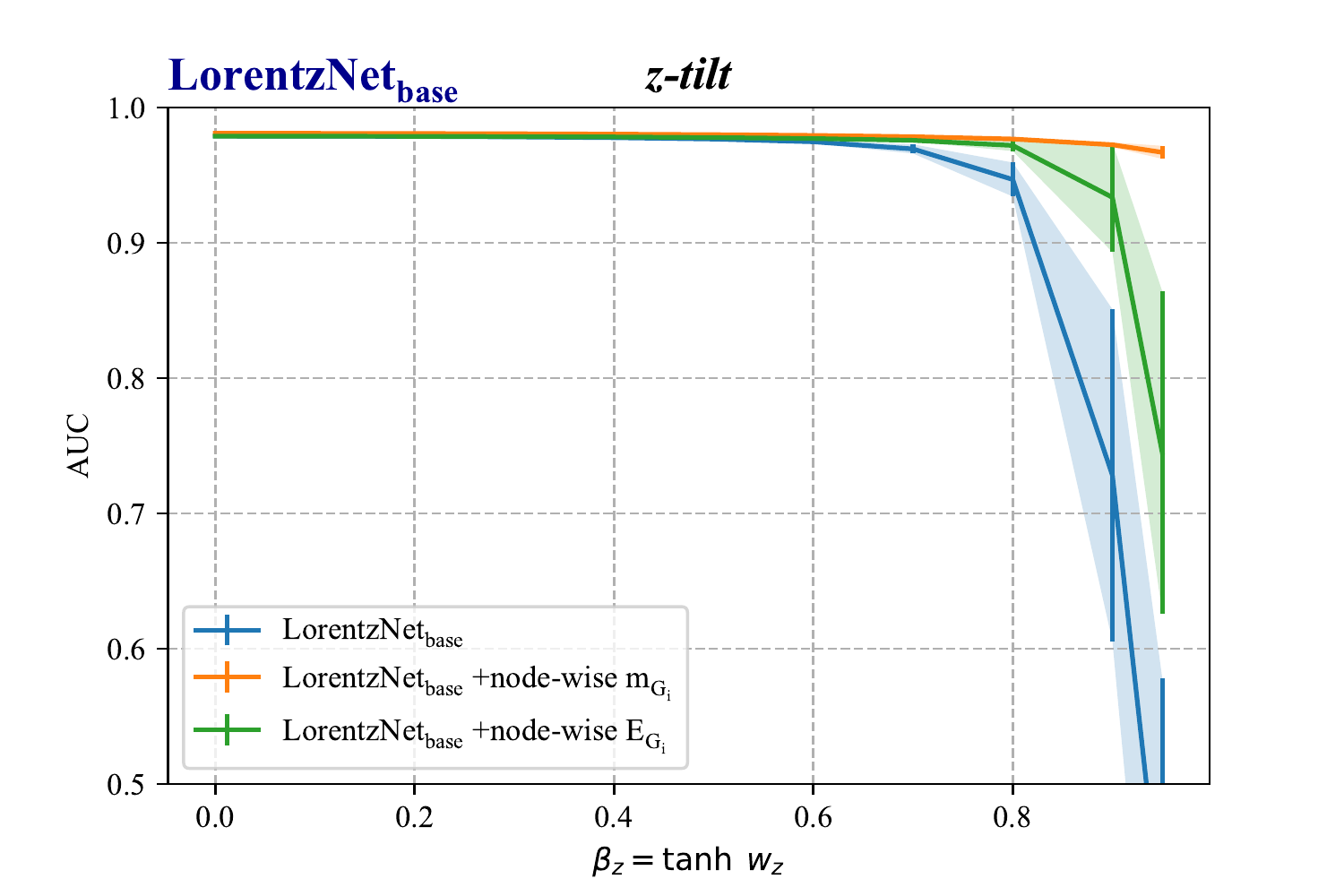} \\
\caption{Network performance in terms of AUC evaluated under various types of Lorentz transformation applied to the test dataset. From left to right: the input jet undergoes a \yz rotation with angle $\alpha_x$,  an \xt boost with rapidity $w_x$, and a $z$-tilt transformation (i.e., a \zt boost with rapidity $w_z$ followed by a \xz rotation to redirect the jet to the $x$ axis). The curves in the plots show different options of nodewise features added to the baseline network. The baselines are chosen as PFN (top), ParticleNet (middle), or \lorentznetbase (bottom). The error bar shows the standard deviation over ten trainings.}
\label{fig:lorentz-attk-node}
\end{center}
\end{figure*}

\begin{figure*}[tb]
\begin{center}
\includegraphics[width=0.32\textwidth]{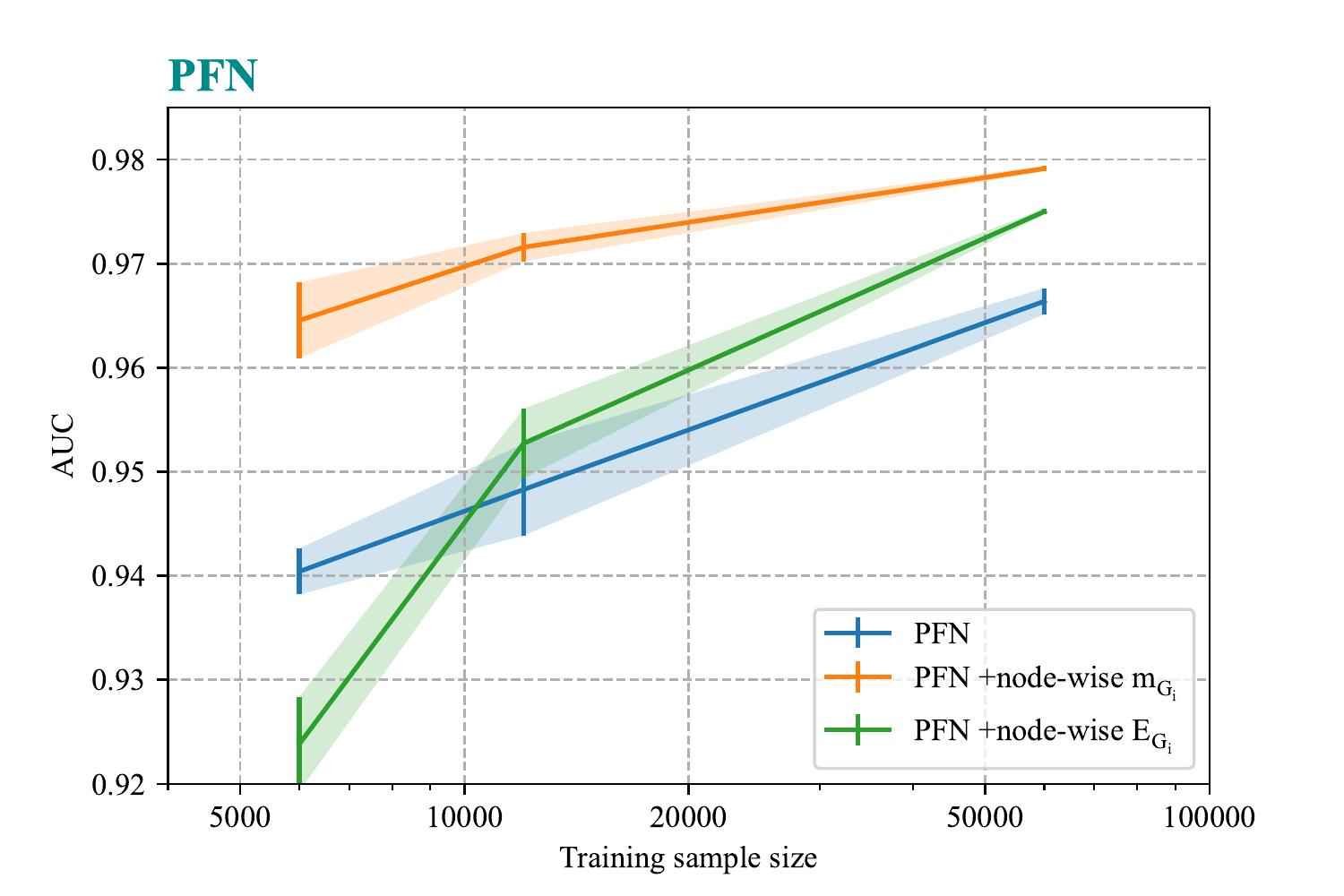}
\includegraphics[width=0.32\textwidth]{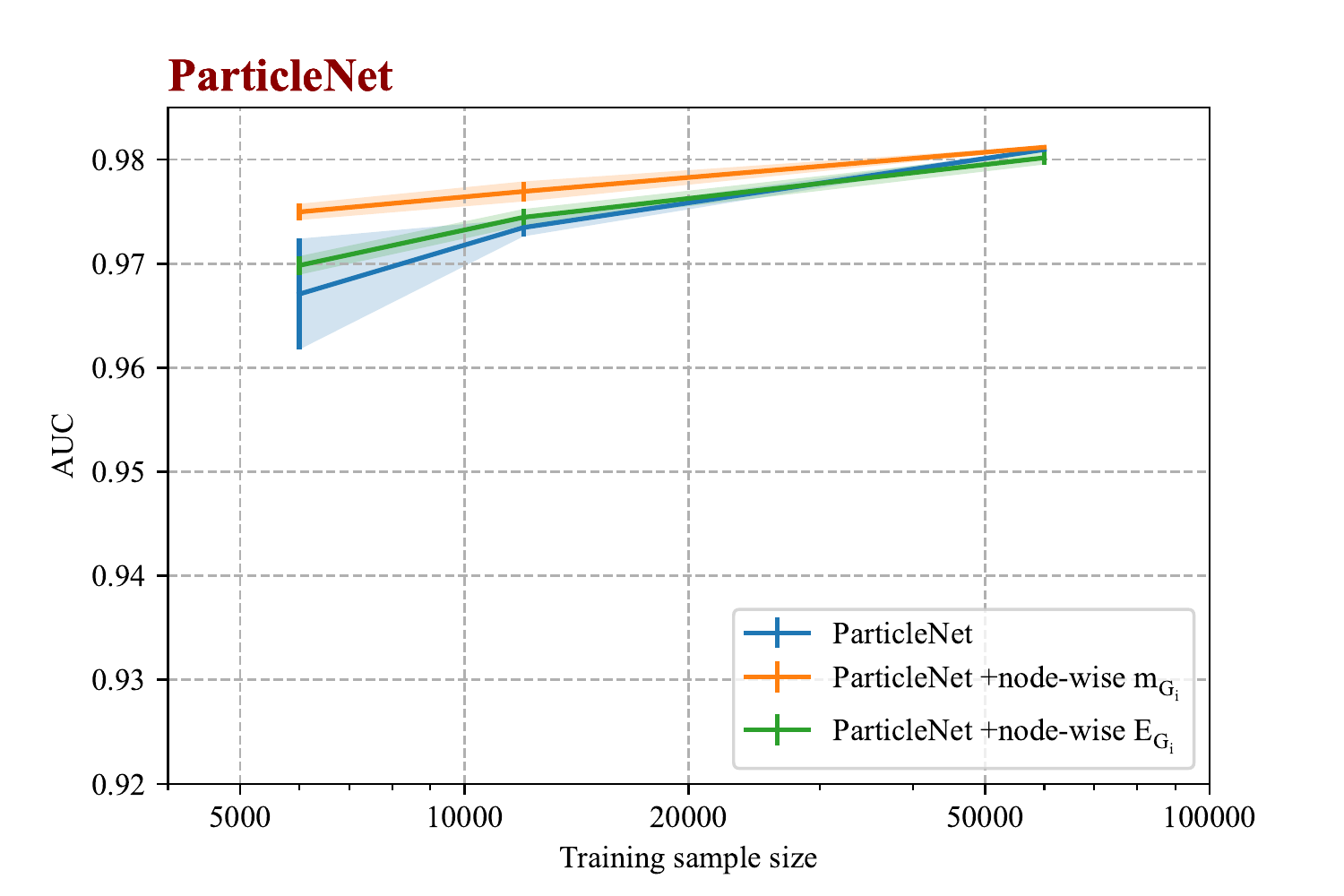}
\includegraphics[width=0.32\textwidth]{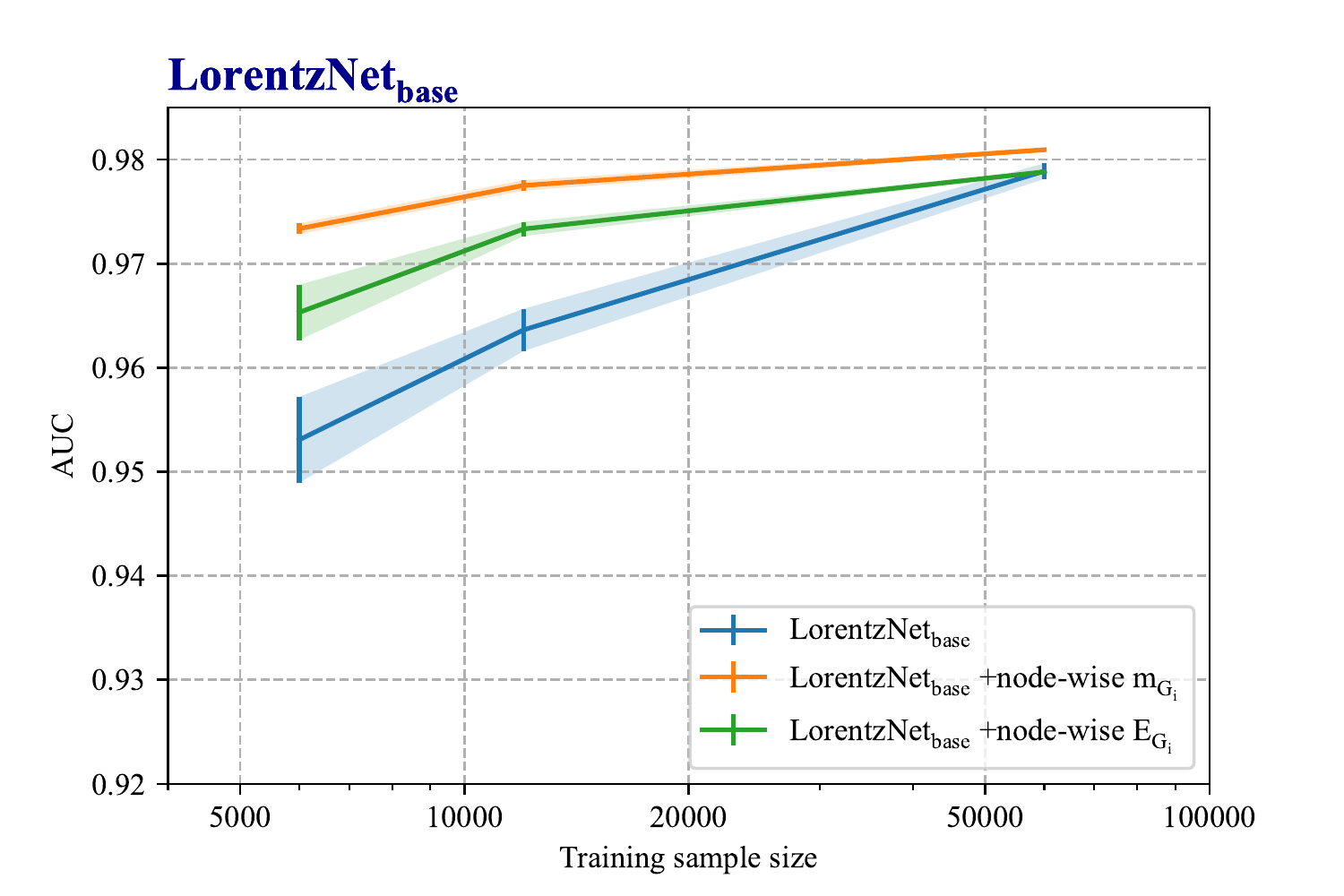} \\
\caption{Network performance in terms of AUC versus the training size. The curves in the plots show different options of nodewise features added to the baseline network. The baselines are chosen as PFN (left), ParticleNet (middle), or \lorentznetbase (right). The error bar shows the standard deviation over ten trainings.}
\label{fig:datasize-node}
\end{center}
\end{figure*}

\begin{itemize}
\item Inclusion of nodewise mass features substantially improves the PFN performance. This demonstrates the huge potential of including manually constructed mass features in improving DNN performance in the era of using low-level inputs.
\item All experiments show a degree of improvement when incorporating mass features. The improvement is relatively small in ParticleNet and \lorentznetbase due to the effectiveness of their plain GNN-based network. However the gain still illustrates that the added Lorentz-symmetry-preserving patch helps improve the network performance.
\item In the case of ParticleNet and \lorentznetbase, the improvement from injecting the nodewise mass features is not as large compared with adding pairwise mass features shown in Table~\ref{tab:pair-perf}. This can be explained from the perspective that, in the case of including node-wide features calculated by Eq.~(\ref{eq:mass-linear-comb}), not all $N(N-1)/2$ Lorentz-invariant features $p_{i}^\mu p_{j,\mu}$ $(\forall\, i,\,j)$ are fed into the network, but only $N$ features composed of their linear combination are taken as the input. In principle, they carry only a part of the information.
\item The behavior under Lorentz boosts and rotations of test dataset and performance trend in using different sample sizes in Figs.~\ref{fig:lorentz-attk-node} and \ref{fig:datasize-node} follows our expectation. This reinforces our conclusion that preserving Lorentz symmetry serves as an inductive bias, a principle that is also applicable in this context.
\end{itemize}

\vspace{8pt}

Finally, Table~\ref{tab:complex} shows the comparison of the model complexity for baseline networks and their variants which incorporate additional features. As can be seen, the effect of the patch for including nodewise features between three baselines is consistent due to the generality of the patch design; the effect of adding pairwise features is rather different for ParticleNet and \lorentznetbase because their patches rely on different mechanisms, as introduced in Sec.~\ref{sec:pair-patch-design}. It is clear from the table that all of our introduced patch structures contain very few parameters compared to the original baselines. It makes the fact even more interesting that the Lorentz invariance property of a very small subnetwork can be successfully reflected onto the entire network. Thus, this finding provides a new angle to argue the important role Lorentz symmetry plays in network design.

\begin{table}[tb]
\setlength{\tabcolsep}{4pt}
\renewcommand{\arraystretch}{1.1}
\caption{The number of trainable parameters and floating point operations (FLOPs) for the three baseline networks and their variants. The ``$+$'' sign indicates the increase in the number with respect to its baseline.}
\label{tab:complex}
\begin{center}
\begin{tabular}{ll|cc}
\hline\hline
Base model  & Variation & \makecell[c]{No. parameters\\($\times 10^{3}$)}   & \makecell[c]{FLOPs\\($\times 10^{6}$)}   \\
\hline
\multirow{2}{*}{PFN} &  ---    & 83.84 & 4.46 \\
    &  +nodewise  & $+26.19$ & $+3.41$ \\
\hline      
\multirow{3}{*}{ParticleNet} &  ---      & 366.16 & 535.73  \\
    &  +pairwise  & $+34.91$ & $+285.29$ \\
    &  +nodewise & $+21.97$ & $+2.83$  \\
\hline
\multirow{3}{*}{LorentzNet\textsubscript{base}} &  ---      &  226.23 & 1997.69 \\
    &  +pairwise  & $+0.43$ & $+7.02$ \\
    &  +nodewise  & $+37.35$ & $+4.8$  \\
\hline\hline
\end{tabular}
\end{center}
\end{table}

\section{Discussion and conclusion}\label{sec:conclusions}

In this work, we study the effect of Lorentz-symmetric design in network performance in a systematic way. We confirm that the answer to the initial question is \textit{yes}: the Lorentz-symmetric design can boost network performance in jet physics, according to our experiments in the context of jet tagging. 

We first find out that the network need not be designed to fully comply with Lorentz symmetry to get the performance boost---only including a substructure invariant to the Lorentz symmetry can ensure a higher performance. Then, inspired by this spirit, we design two patches that can be generally used to improve the network performance.

\begin{itemize}
\item First, the pairwise mass feature can be injected into a GNN-based model, e.g., ParticleNet and LorentzNet, in their intrinsically supported way to assist in building the edge features of the graph that participate in the message-passing mechanism.
\item Second, as a more universal solution, we propose the design of the ``nodewise mass'' feature, which is constructed by the invariant mass of various friend particles of a given particle, and propose a general patch structure for injecting the feature into the primary network structure block by block.
\end{itemize}

We conduct experiments on PFN, ParticleNet, and the weakened version of LorentzNet and see general improvements when incorporating these mass features in two different ways. We use Lorentz boosts and rotation experiments to illustrate that the underlying symmetry preservation plays a role in the network training to achieve higher performance. Especially, we design a specific experiment to introduce the patch network structure adhering to various levels of Lorentz subsymmetries. The results indicate improved performance with the incorporation of more symmetry levels. This finding demonstrates that respecting \textit{full} Lorentz symmetry is particularly beneficial, aiding the network in achieving higher performance. This goes beyond the more commonly held belief in our community that symmetries related to the boosts along the beamline (\zt boost) and azimuthal rotation (\xy rotation) are the primary ones to be integrated into the design of jet neural networks. We then find that injecting mass features in two ways improves the network performance, especially when trained on a small training sample. This further demonstrates that Lorentz symmetry preservation is an effective way to assist the network in achieving higher performance, hence a real but often overlooked ``inductive bias'' in the jet physics task.

From another perspective, this work makes a successful step forward in understanding the interpretability of neural networks, in terms of how the networks incorporate symmetries using the dedicated variables we inject into the network. We show to the community that the previously discovered pairwise mass features, which are capable of improving network performance, find their root in the incorporation of full Lorentz symmetry in the network's substructure to process these variables.

\section{Outlook}\label{sec:outlook}

This work reveals, in the context of jet tagging, that Lorentz symmetry is an inductive bias, which, by properly hinting to the network, can enhance the network performance. Hence, one of the primary goals of our work is to draw attention to such inductive biases in future jet network designs. In this work, we propose the nodewise mass recipe, which is more general and capable of being applied to a variety of networks; however, we also emphasize that, with the goal of achieving state-of-the-art performance, it is more necessary to utilize the pairwise mass feature, as it contains more abundant Lorentz invariance properties inside a jet, and to incorporate it with advanced baseline networks, which can be either GNNs or attention-based models like Transformers. We note that both LorentzNet~\cite{Gong:2022lye} and ParT~\cite{pmlr-v162-qu22b} have adopted the pairwise mass design. This also explains to some extent the high performance they have exhibited.

Beyond the jet tagging task, it is interesting to study the effect of applying the patch structures in other physics scenarios that treat jets as a point cloud (set) of particles, for instance, in the regression of jet properties~\cite{Qiu:2022xvr}, in the jet assignment tasks~\cite{Fenton:2020woz,Lee:2020qil,Badea:2022dzb}, and in the generation task of jets with use of a generative model~\cite{kansal2021particle}. Furthermore, tasks that process whole collision events instead of a single jet may also draw on such patches in the network design. A typical example includes using a variational autoencoder to identify anomalous events in the search for new physics~\cite{Ostdiek:2021bem}.

This perspective broadens to a potentially more promising viewpoint.
For deep learning tasks using more primitive data as input, e.g., the raw data collected in calorimeters, which deposit energies in the regular grid or the data from the tracker storing the hit information~\cite{Pata:2021oez,DeZoort:2021rbj}, a key fact remains, i.e., the essence of these data lies in the information of outgoing particles. Therefore, we conjecture that Lorentz symmetry is equally important for such tasks. Special designs of the Lorentz-symmetry-preserving network to adapt these sources of input can be an interesting field for future study.

In addition to the points discussed above, we would also like to address that, for a better understanding of the role that symmetry-preservation plays in the network performance, there are yet room and means. Regarding the systematic study of Lorentz-symmetric design, this work adopts a universal paradigm, focusing on a segment of the network (a patch structure) and ensuring its invariance under Lorentz symmetry or its subsymmetries. However, this approach does not include the equivariant case, as such design can be more specialized.
While integrating this case into our general study presents challenges, we think that the Lorentz-equivariant designs may still inspire the next generation of high-performing networks. Hence, we also emphasize the importance of these designs in future research.
Additionally, although the mass has manifested itself in our study as a symmetry-relevant feature intrinsic in jet physics, when we focus on the heavy resonance jet tagging task, mass is also a direct signature to distinguish a specific type of jets or subjets. It would be interesting to study the role of masses and their symmetry-preserving property in other scenarios, e.g., the jet flavor tagging task, where jets cannot be explicitly distinguished by the mass variable itself. This will be more helpful to understand the role of mass in the network.

\acknowledgements
The work of C.~L., S.~Q., and Q.~L. is supported by National Natural Science Foundation of China under Grants No.~12061141002 and No.~2075004. C.~L. and Q.~L. thank 	
Raghav Kansal, Cristina Mantilla Suarez, Javier Duarte, and Zhengyun You for the helpful discussions on various jet tagging aspects. H.~Q. is thankful for the discussion with Barry Dillon, Anja Butter, and Tilman Plehn. C.~L., H.~Q., and S.~Q. are thankful for the discussion with Alexandre De Moor, Loukas Gouskos, Raffaele Gerosa, Stephane Cooperstein, and 
Santeri Laurila. C.~L. is thankful for the support from Yannan Pan and the discussion with Ziyang Zhang and Jose M. Munoz.

\appendix

\bibliographystyle{apsrev4-1}
\bibliography{main}

\section{Supplemented proofs}\label{sec:proof}
This appendix provides proofs of several pairwise feature properties under Lorentz transformations which are discussed in the context.

\begin{proposition}~\label{prop:1}
In the limit of $y_{y,i},\,y_{z,i}\sim o(1)$ for all particles $i$, where $y_y = \arctanh (p_y/E)$ and $y_z = \arctanh (p_z/E)$, considering all particles as massless, the pairwise feature $p_{ij} = \Delta R_{ij} (p_{{\rm T},i} + p_{{\rm T},j})$ for particles $i$ and $j$ is an approximate invariance to the boost transformation on the $x$ axis.
\end{proposition}

\begin{proof}
In the massless case and under the limit of $y_{y,i},\,y_{z,i}\sim o(1)$, from Eq.~(\ref{eq:etaphiapp}) we have
\begin{align*}
\begin{split}
    \eta &= y_y, \\
    \phi &= y_z + o(y_y, y_z).
\end{split}
\end{align*}
Hence,
\begin{align*}
\begin{split}
    \Delta R_{ij} &= \sqrt{(\eta_i-\eta_j) ^2 + (\phi_i - \phi_j)^2} \\
      &= \sqrt{(y_{y,i} - y_{y,j})^2 + (y_{z,i} - y_{z,j})^2} + o(y_y, y_z).
\end{split}
\end{align*}
For a boost on the $x$ axis with rapidity $w$, for each particle, given $p_x = E (1-y_y^2 - y_z^2)^\frac{1}{2} = E + o(1)$,
we have
\begin{align*}
\begin{split}
    E' &= E\;\cosh w + p_x\;\sinh w= e^w E + o(1), \\
\end{split}
\end{align*}
$p_x' = e^w E + o(1)$, $p_y'=p_y$ and $p_z'=p_z$. Hence, we have
\begin{align*}
\begin{split}
    y_y' &= e^{-w} y_y + o(y_y,y_z), \\
    y_z' &= e^{-w} y_z + o(y_y,y_z).
\end{split}
\end{align*}
Therefore, for each particle $i$ or pair $(i,j)$, after transformation,
\begin{align*}
\begin{split}
    p_{{\rm T},i}' &= e^w \;p_{{\rm T},i} + o(1), \\
    \Delta R_{ij}' &= e^{-w} \;\Delta R_{ij} + o(y_{y,i}, y_{z,i}),
\end{split}
\end{align*}
Thus,
\begin{align*}
\begin{split}
    \Delta R_{ij}'(p_{{\rm T},i}' + p_{{\rm T},j}') = R_{ij}(p_{{\rm T},i}+ p_{{\rm T},j}) + o(y_{y,i}, y_{z,i})
\end{split}
\end{align*}
which is invariant in the first-order $y_{y,i},\,y_{z,i}\sim o(1)$ limit.
\end{proof}

\begin{proposition}
In the limit of $y_{y,i},\,y_{z,i}\sim o(1)$ for all particles $i$ and considering all particles as massless, the pairwise squared mass $m_{ij}^2 = p_{i}^\mu p_{j,\mu}$ is equal to $\frac{1}{2}\Delta R_{ij}^2 p_{{\rm T},i}\,p_{{\rm T},j}$ to the leading order.
\end{proposition}
\begin{proof}
According to Proposition~\ref{prop:1}, and given
\begin{align*}
\begin{split}
m_{ij} ^2 = E_iE_j - p_{x,i}p_{x,j} - p_{y,i}p_{y,j} - p_{z,i}p_{z,j},
\end{split}
\end{align*}
we have
\begin{align*}
\begin{split}
m_{ij} ^2 &= E_iE_j \\
 & \quad\quad - E_iE_j\Big(1-\frac{1}{2}y_{y,i}^2-\frac{1}{2}y_{z,i}^2\Big)\Big(1-\frac{1}{2}y_{y,j}^2-\frac{1}{2}y_{z,j}^2\Big) \\
 & \quad\quad - E_iE_jy_{y,i}y_{y,j} - E_iE_jy_{z,i}y_{z,j} + o(y^2) \\
 &= E_iE_j\Big(\frac{1}{2}y_{y,i}^2+\frac{1}{2}y_{z,i}^2 + \frac{1}{2}y_{y,j}^2+\frac{1}{2}y_{z,j}^2 \\
 & \quad\quad - y_{y,i}y_{y,j} - y_{z,i}y_{z,j}\Big) + o(y^2) \\
 &=\frac{1}{2}E_i E_j \Delta R_{ij}^2+ o(y^2) \\
 &=\frac{1}{2}\Delta R_{ij}^2 p_{{\rm T},i}\,p_{{\rm T},j}+ o(y^2).
\end{split}
\end{align*}
\end{proof}

\section{Training setup}\label{sec:config}
The training of PFN, ParticleNet, and LorentzNet models and their variants to incorporate pairwise or nodewise features is performed on an Nvidia RTX 3090 GPU. 
The original models of PFN and ParticleNet are taken from the public version provided in Ref.~\cite{pmlr-v162-qu22b} and the LorentzNet model is taken from Ref.~\cite{Gong:2022lye}.
The batch size is set as 512 for PFN, 512 for ParticleNet (128 for ParticleNet incorporating pairwise features), and 256 for LorentzNet. The initial learning rate (LR) is set as $2.5\times 10^{-3}$ for PFN, $1\times 10^{-2}$ for ParticleNet ($2.5\times 10^{-3}$ for ParticleNet incorporating pairwise features), and $2\times 10^{-3}$ for LorentzNet.

In our experiments using the top tagging dataset, the training is performed on a total of 20 epochs, with an epoch defined as a whole iteration over the used dataset. For our case where only a portion of training jets are used, they are selected by the event number \texttt{event\_no}, e.g., the $60\,000$ jets used for training are selected by \texttt{event\_no \% 20 == 0} from the training dataset.

In experiments using the JetClass dataset, when the training sample size is \{$60\,000$, $200\,000$, $500\,000$\}, the training is performed on 20 epochs, with an epoch defined as a whole iteration over the used dataset. In order to make the training controllable when the data size grows further, a different configuration is adopted. When the training size is \{$2\times 10^{6}$, $10\times 10^{6}$, $100\times 10^{6}$\}, the number of training epochs is 50, where each epoch is defined as iterating \{0.5, 0.25, 0.1\} of the selected training dataset. This allows the configuration to be the same as in Ref.~\cite{pmlr-v162-qu22b} when training on the full $100\times 10^{6}$ dataset.

For all the training, we employ the same optimizer and the scheduler as in Ref.~\cite{pmlr-v162-qu22b}. The Lookahead optimizer~\cite{zhang2019lookahead} with $k = 6$ and $\alpha = 0.5$ is used to minimize the cross-entropy loss. The inner optimizer is RAdam~\cite{Liu2020On} with $y_1 = 0.95$, $y_2 = 0.999$, and $\epsilon = 10^{-5}$.
For our training schedule, the LR remains constant for the first 70\% of the iterations and then decays exponentially, changes at beginning of every following epoch, down to 1\% of the initial value at the end of the training. In the training of LorentzNet and its variants, a weight decay of 0.01 is adopted according to Ref.~\cite{Gong:2022lye}.

\end{document}